\documentclass[11pt]{article}

\usepackage{amssymb,amsfonts,amsmath,amsthm,url, placeins, amsthm}
\usepackage{multirow,ctable} 
\usepackage{paralist,xspace}
\usepackage[text={6.75in,9.5in},centering]{geometry}
\usepackage{epsfig,graphicx,wrapfig}

\usepackage{multicol}
\usepackage{color}

\definecolor{gold}{rgb}{0.85,.66,0}
\definecolor{cherry}{rgb}{0.9,.1,.2}
\definecolor{burgundy}{rgb}{0.8,.2,.2}
\definecolor{orangered}{rgb}{0.85,.3,0}
\definecolor{orange}{rgb}{0.85,.4,0}
\definecolor{olive}{rgb}{.45,.4,0}
\definecolor{lime}{rgb}{.6,.9,0}
\definecolor{green}{rgb}{.2,.7,0}
\definecolor{darkgreen}{rgb}{.1,.5,0}
\definecolor{grey}{rgb}{.4,.4,.2}
\definecolor{brown}{rgb}{.4,.2,.1}
\definecolor{blue}{rgb}{0,.0, .81}
\definecolor{bluepurple}{rgb}{.3, .0, .7}

\theoremstyle{definition}

\newtheorem{theorem}{Theorem}[section]
\newtheorem*{theorem*}{Theorem}
\newtheorem{corollary}[theorem]{Corollary}

\newtheorem{example}[theorem]{Example}
\newtheorem*{example*}{Example}
\newtheorem{proposition}[theorem]{Proposition}
\newtheorem{definition}[theorem]{Definition}
\newtheorem{lemma}[theorem]{Lemma}

\def\R{\mathbb{R}}
\def\C{\mathcal{C}}
\def\od{\stackrel{\mathrm{def}}{=}}
\def\U{\mathcal{U}}
\def\N{\mathcal{N}}
\def\M{\mathcal{M}}

\def\RR{\mathbb{R}}

\def\Lk{\operatorname{Lk}}
\def\RF{\operatorname{RF}}

\def\ZZ{\mathbb{Z}}
\def\kk{\mathbf{k}}

\begin{document}
\begin{center}
\textbf{ \Large What makes a neural code convex?}
\medskip

Carina Curto$^1$, Elizabeth Gross$^2$, Jack Jeffries$^{3,4}$, Katherine Morrison$^{1,5}$, Mohamed Omar$^6$,\\ 
Zvi Rosen$^{1,7,8}$, Anne Shiu$^9$, and Nora Youngs$^{6,10}$\\
Dec 16, 2016 \\
\end{center}

\begin{small}
\hspace{.25in} $^1$ Department of Mathematics, The Pennsylvania State University, University Park, PA 16802

\hspace{.25in} $^2$ Department of Mathematics, San Jos\'{e} State University, San Jos\'{e}, CA 95192

\hspace{.25in} $^3$ Department of Mathematics, University of Utah, Salt Lake City, UT 84112

\hspace{.25in} $^4$ Department of Mathematics, University of Michigan, Ann Arbor, MI 48109

\hspace{.25in} $^5$ School of Mathematical Sciences, University of Northern Colorado, Greeley, CO 80639 

\hspace{.25in} $^6$ Department of Mathematics, Harvey Mudd College, Claremont, CA 91711

\hspace{.25in} $^7$ Department of Mathematics, University of California, Berkeley, Berkeley, CA 94720

\hspace{.25in} $^8$ Department of Mathematics, University of Pennsylvania, Philadelphia, PA 19104

\hspace{.25in} $^9$ Department of Mathematics, Texas A\&M University, College Station, TX 77843

\hspace{.25in} $^{10}$Department of Mathematics, Colby College, Waterville, ME 04901
\end{small}

\section*{Abstract}
Neural codes allow the brain to represent, process, and store information about the world. 
Combinatorial codes, comprised of binary patterns of neural activity, encode information via the collective behavior of populations of neurons.  
A code is called {\it convex} if its codewords correspond to regions defined by an arrangement of convex open sets in Euclidean space.  
Convex codes have been observed experimentally in many brain areas, including sensory cortices and the hippocampus, where neurons exhibit convex receptive fields. 
What makes a neural code convex?  That is, how can we tell from the intrinsic structure of a code if there exists a corresponding arrangement of convex open sets?  
In this work, we provide a complete characterization of {\it local obstructions} to convexity.  This motivates us to define {\it max intersection-complete codes}, a family guaranteed to have no local obstructions.  We then show how our characterization enables one to use free resolutions of Stanley-Reisner ideals in order to detect violations of convexity.  Taken together, these results provide a significant advance in understanding the intrinsic combinatorial properties of convex codes.

\medskip
\noindent \textbf{Keywords:} neural coding, convex codes, simplicial complex, link, Nerve lemma, Hochster's formula

\section{Introduction}\label{sec:intro}

Cracking the neural code is one of the central challenges of neuroscience. Typically, this has been understood as finding the relationship between the activity of neurons and the stimuli they represent.  To uncover the principles of neural coding, however, it is not enough to describe the various mappings between stimulus and response.  One must also understand the intrinsic structure of neural codes, independently of what is being encoded \cite{neural_ring}.  

Here we focus our attention on {\it convex codes}, which are comprised of activity patterns for neurons with classical receptive fields.  A {\it receptive field} $U_i$ is the subset of stimuli that induces neuron $i$ to respond.  Often, $U_i$ is a convex subset of some Euclidean space (see Figure~\ref{fig:place-fields}).  A collection of convex sets $U_1,\ldots,U_n \subset \RR^d$ naturally associates to each point $x \in \RR^d$ a binary response pattern, $c_1\cdots c_n \in \{0,1\}^n$, where $c_i = 1$ if $x \in U_i$, and $c_i = 0$ otherwise.  The set of all such response patterns is a convex code.  
 
Convex codes have been observed experimentally in many brain areas, including sensory cortices and the hippocampus.  Hubel and Wiesel's discovery in 1959 of orientation-tuned neurons in the primary visual cortex was perhaps the first example of convex coding in the brain \cite{HubelWiesel59}.  This was followed by O'Keefe's discovery of hippocampal place cells in 1971 \cite{ OKeefe}, showing that convex codes are also used in the brain's representation of space.  Both discoveries were groundbreaking for neuroscience, and were later recognized with Nobel Prizes in 1981 \cite{nobelprize81} and 2014 \cite{nobelprize14}, respectively.

Our motivating examples of convex codes are, in fact, hippocampal place cell codes.
A {\it place cell} is a neuron that acts as a position sensor, exhibiting a high firing rate when the animal's location lies inside the cell's preferred region of the environment -- its {\it place field}.  Figure~\ref{fig:place-fields} displays the place fields of four place cells recorded while a rat explored a two-dimensional environment.  Each place field is an approximately convex subset of $\RR^2$.  Taken together, the set of all neural response patterns that can arise in a population of place cells comprise a convex code for the animal's position in space.  Note that in the neuroscience literature, convex receptive fields are typically referred to as {\it unimodal}, emphasizing the presence of just one ``hot spot'' in a stimulus-response heat map such as those depicted in Figure~\ref{fig:place-fields}.
\medskip

\begin{figure}[!h]
\begin{center}
\includegraphics[width=5in]{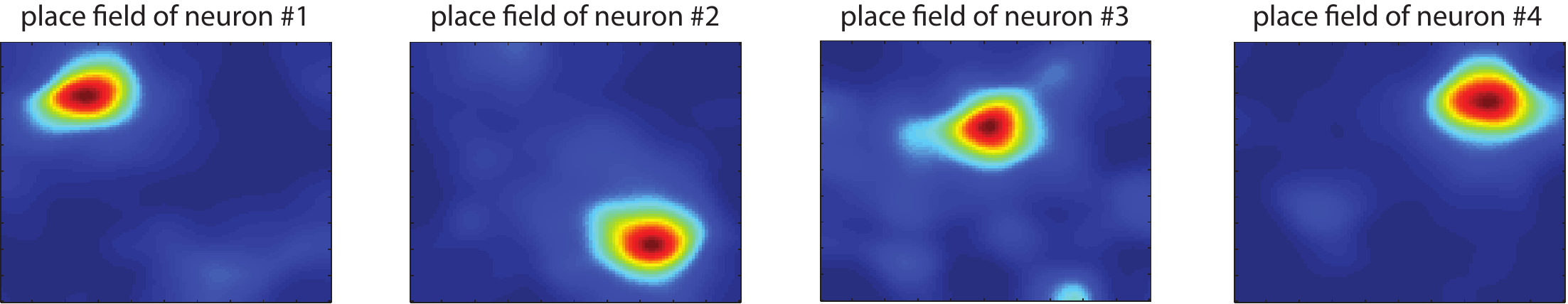}
\end{center}
\caption{Place fields of four CA1 pyramidal neurons (place cells) in rat hippocampus, recorded while the animal explored a 1.5m $\times$ 1.5m square box environment.  Red areas correspond to regions of space where the corresponding place cells exhibited high firing rates, while dark blue denotes near-zero activity.  
Place fields were computed from data provided by the Pastalkova lab, as described in \cite{clique-topology}.}
\label{fig:place-fields}
\end{figure}

%
%

Despite their relevance to neuroscience, the mathematical theory of convex codes was initiated only recently \cite{neural_ring, neuro-coding}.  In particular, it is not clear what are the intrinsic combinatorial signatures of convex and non-convex codes.  Identifying such features will enable us to infer coding properties from population recordings of neural activity, without needing {\it a priori} knowledge of the stimuli being encoded.   This may be particularly important for studying systems like olfaction, where the underlying ``olfactory space" is potentially high-dimensional and poorly understood.  Having intrinsic signatures of convex codes is a critical step towards understanding whether something like convex coding may be going on in such areas.  Understanding the structure of convex codes is also essential to uncovering the basic principles of how neural networks are organized in order to learn, store, and process information.

\subsection{Convex codes}

By a {\it neural code}, or simply a {\it code}, on $n$ neurons we mean a collection of binary strings $\C \subseteq \{0,1\}^n.$  The elements of a code are called {\it codewords}.  We interpret each binary digit as the ``on'' or ``off'' state of a neuron, and consider 0/1 strings of length $n$ and subsets of $[n] = \{1,\ldots,n\}$ interchangeably.  For example, $1101$ and $0100$ are also denoted $\{1,2,4\}$ and $\{2\}$, respectively.
We will always assume $00\cdots0 \in \C$ (i.e., $\emptyset \in \C$); this assumption simplifies notation in various places but does not alter the core results.

Let $X$ be a topological space, and consider a collection of open sets $\U = \{U_1,\ldots,U_n\}$, where $\bigcup_{i = 1}^n U_i \subsetneq X$.  Any such $\U$ defines a code,
\begin{equation}\label{eq:C(U)}
\C(\U) \od \left\{ \sigma \subseteq [n] \mid U_\sigma \setminus \bigcup_{j \in [n] \setminus \sigma} U_j \neq \emptyset \right\},
\end{equation}
where $U_\sigma \od \cap_{i \in \sigma} U_i$ for $\sigma \subseteq [n].$ 
In other words, each codeword $\sigma$ in $\C(\U)$ corresponds to the portion of $U_\sigma$ that
is not covered by other sets.  In particular, $\C(\U)$ is not the same as the {\em nerve} of the cover, $\N(\U)$ (see Section~\ref{sec:nerve}).  By convention, $U_\emptyset = X$, and so $\emptyset \in \C(\U)$.

For any code $\C$, there always exists an open cover $\U$ such that $\C = \C(\U)$ \cite[Lemma 2.1]{neural_ring}.  However, it may be impossible to choose the $U_i$s to all be convex.  We thus have the following definitions, which were first introduced in \cite{neural_ring}:

\begin{definition} \label{def:convex}
Let $\C$ be a binary code on $n$ neurons.  If there exists a collection of open sets $\U = \{U_1,\ldots,U_n\}$ such that $\C = \C(\U)$ and the $U_i$s are all convex subsets of $\RR^d$, then we say that $\C$ is a {\em convex} code. The smallest $d$ such that this is possible is called the {\em minimal embedding dimension}, $d(\C)$. 
\end{definition}

\noindent Note that the definition of a convex code is extrinsic: a code is convex if it can be realized by an arrangement of convex open sets in some Euclidean space.  How can we characterize convex codes {\it intrinsically}?  If a code is {\it not} convex, how can we prove this?  If a code {\it is} convex, what is the minimal dimension needed for the corresponding open sets? 

In this work, we tackle these questions by building on mathematical ideas from \cite{neural_ring} and \cite{no-go}.  In particular, we study {\it local obstructions} to convexity, a notion first introduced in \cite{no-go}.
Our main result is Theorem~\ref{thm:Thm1}, which provides a complete characterization of codes with no local obstructions.  In Section~\ref{sec:examples} we present a series of examples that illustrate the ideas summarized in Section~\ref{sec:main-results}.
Sections~\ref{sec:local-obstructions} and~\ref{sec:mandatory-codewords} are devoted to additional background and technical results needed for the proof of Theorem~\ref{thm:Thm1}.
Finally, in Section~\ref{sec:hochster} we show how tools from combinatorial commutative algebra, such as Hochster's formula, can be used to determine that a code is not convex. 

\subsection{Preliminaries}\label{sec:prelim}

\paragraph{Simplicial complexes.}
A {\it simplicial complex} $\Delta$ on $n$ vertices is a nonempty collection of subsets of $[n]$ that is closed under inclusion.  In other words, if $\sigma \in \Delta$ and $\tau \subset \sigma$, then $\tau \in \Delta$.   The elements of a simplicial complex are called {\it simplices} or {\it faces}. The {\it dimension} of a face, $\sigma \in \Delta,$ is defined to be $|\sigma| -1$.  The dimension of a simplicial complex $\Delta$ is equal to the dimension of its largest face: $\max_{\sigma \in \Delta} |\sigma|-1$.   If $\Delta$ consists of all $2^n$ subsets of $[n]$, then it is the {\it full simplex} of dimension $n-1$.  The {\it hollow simplex} contains all proper subsets of $[n]$, but not $[n]$, and thus has dimension $n-2$. 

Faces of a simplicial complex that are maximal under inclusion are referred to as \emph{facets}.    If we consider the facets together with all their intersections, we obtain the set
$$\mathcal{F}_{\cap}(\Delta) \od \left\{\bigcap_{i=1}^k \rho_i \mid \rho_i \text{ is a facet  of } \Delta \text{ for each } i = 1,\ldots,k \right\} \cup \{\emptyset\}.$$
We refer to the elements of $\mathcal{F}_{\cap}(\Delta)$ as {\it max intersections} of $\Delta$.
The empty set is added so that $\mathcal{F}_{\cap}(\Delta)$ can be regarded as a code, consistent with our convention that the all-zeros word is always included.  

Restrictions and links are standard constructions from simplicial complexes.  
The {\em restriction} of $\Delta$ to $\sigma$ is the simplicial complex
\[
 \Delta | _ {\sigma} \od \{\omega \in \Delta \mid \omega \subseteq \sigma \}\,.
 \]
For any $\sigma \in \Delta$, the {\em link} of $\sigma$ inside $\Delta$ is 
\[\Lk_{\sigma} (\Delta) \od \{ \omega \in \Delta \mid \sigma \cap \omega = \emptyset \text{ and } \sigma \cup \omega \in \Delta \}\,.
\]
Note that it is more common to write $\Lk_\Delta(\sigma)$ or $\operatorname{link}_\Delta(\sigma)$, instead of $\Lk_\sigma(\Delta)$ (see, for example, \cite{miller-sturmfels}).  However, because we will often fix $\sigma$ and consider its link inside different simplicial complexes, such as $\Delta|_{\sigma \cup \tau}$, it is more convenient to put $\sigma$ in the subscript.

\paragraph{The simplicial complex of a code.}
To any code $\C$, we can associate a simplicial complex $\Delta(\C)$ by simply including all subsets of codewords:
$$\Delta(\C) \od \left\{\sigma \subseteq [n] \mid \sigma \subseteq c \text{ for some } c \in \C \right\}.$$
$\Delta(\C)$ is called the {\it simplicial complex of the code}, and is the smallest simplicial complex that contains all elements of $\C$.   The facets of $\Delta(\C)$ correspond to the codewords in $\C$ that are maximal under inclusion: these are the {\it maximal codewords}.

\paragraph{Local obstructions to convexity.}
At first glance, it may seem that all codes should be convex, since the convex sets $U_i$ can be chosen to reside in arbitrarily high dimensions.  This is not the case, however, as non-convex codes arise for as few as $n=3$ neurons \cite{neural_ring}.  To understand what can go wrong, consider a code with the following property: any codeword with a 1 in the first position also has a 1 in the second or third position, but no codeword has a 1 in all three positions.  This implies that any corresponding cover $\U$ must have $U_1 \subseteq U_2 \cup U_3$, but $U_1 \cap U_2 \cap U_3 = \emptyset$.  The result is that $U_1$ is a disjoint union of two nonempty open sets, $U_1 \cap U_2$ and $U_1 \cap U_3$, and is hence disconnected.  Since all convex sets are connected, we conclude that our code cannot be convex.  The contradiction stems from a topological inconsistency that emerges if the code is assumed to be convex.

This type of topological obstruction to convexity generalizes to a family of {\it local obstructions}, first introduced in \cite{no-go}.  We define local obstructions precisely in Section~\ref{sec:local-obstructions}.  
There we also show that a code with one or more local obstructions cannot be convex:

\begin{lemma}\label{lemma1}
If $\C$ is a convex code, then $\C$ has no local obstructions.
\end{lemma}

\noindent This fact was first observed in \cite{no-go}, using slightly different language.  The converse, unfortunately, is not true.  See Example~\ref{ex:counterexample} for a counterexample that first appeared in \cite{counterexample}.
  
\subsection{Summary of main results}\label{sec:main-results}

To prove that a neural code is convex, it suffices to exhibit a convex realization.  That is, it suffices to find a set of convex open sets $\U = \{U_1,\ldots,U_n\}$ such that $\C = \C(\U)$.
Our strategy for proving that a code is {\it not} convex is to show that it has a local obstruction to convexity.  Which codes have local obstructions?

Perhaps surprisingly, the question of whether or not a given code $\C$ has a local obstruction can be reduced to the question of whether or not it contains a certain minimal code, $\C_{\min}(\Delta)$, which depends only on the simplicial complex $\Delta = \Delta(\C)$.  This is our main result:

\begin{theorem}[Characterization of local obstructions]\label{thm:Thm1}
For each simplicial complex $\Delta$, there is a unique minimal code $\C_{\min}(\Delta)$ with the following properties:
\begin{itemize} 
\item[(i)] the simplicial complex of $\C_{\min}(\Delta)$ is $\Delta$, and
\item[(ii)] for any code $\C$ with simplicial complex $\Delta$, $\C$ has no local obstructions if and only if $\C \supseteq \C_{\min}(\Delta)$.  
\end{itemize}
Moreover, $\C_{\min}(\Delta)$ depends only on the topology of the links of $\Delta$:
$$\C_{\min}(\Delta) = \{\sigma \in \Delta \mid \Lk_\sigma(\Delta) \text{ is non-contractible}\} \cup \{\emptyset\}.$$
\end{theorem}
\medskip

 We will regard the elements of $\C_{\min}(\Delta)$ as ``mandatory'' codewords with respect to convexity, because they must all be included  in order for a code $\C$ with $\Delta(\C) = \Delta$ to be convex.
From the above description of $\C_{\min}(\Delta)$, we can prove the following lemma.

\begin{lemma}\label{lemma:key-lemma}
$\C_{\min}(\Delta) \subseteq \mathcal{F}_{\cap}(\Delta).$  That is, 
each nonempty element of $\C_{\min}(\Delta)$ is an intersection of facets of $\Delta$.  
\end{lemma}

Our proofs of Theorem~\ref{thm:Thm1} and Lemma~\ref{lemma:key-lemma} are given in Section~\ref{sec:Thm1}.
Unfortunately, finding all elements of $\C_{\min}(\Delta)$ for arbitrary $\Delta$ is, in general, undecidable (see Section~\ref{sec:hochster}).  Nevertheless, we can algorithmically compute a subset, $\M_H(\Delta),$ of  
``homologically-detectable'' mandatory codewords.
In Section~\ref{sec:hochster} we show how to compute $\M_H(\Delta)$ 
using machinery from combinatorial commutative algebra.  
Lemma~\ref{lemma:key-lemma} also tells us that every element of $\C_{\min}(\Delta)$ must be an intersection of facets of $\Delta$ -- that is, an element of $\mathcal{F}_{\cap}(\Delta)$.
We thus have the inclusions:
\begin{equation}\label{eq:inclusions}
\M_H(\Delta) \subseteq \C_{\min}(\Delta) \subseteq \mathcal{F}_{\cap}(\Delta),
\end{equation}
where both $\M_H(\Delta)$ and $\mathcal{F}_{\cap}(\Delta)$ are straightforwardly computable.   Note that if $\M_H(\Delta) = \mathcal{F}_{\cap}(\Delta)$, then we can conclude that $\C_{\min}(\Delta) = \mathcal{F}_{\cap}(\Delta).$  Moreover, 
 if $\C \supseteq \mathcal{F}_{\cap}(\Delta)$, then $\C \supseteq  \C_{\min}(\Delta)$, and thus $\C$  has no local obstructions (and is potentially convex).  This motivates the following definition:
 
\begin{definition}\label{def:max-int-complete}
A neural code $\C$ is {\em max $\cap$-complete} (or max intersection-complete) if $\C \supseteq \mathcal{F}_{\cap}(\Delta(\C))$.
\end{definition}
 
We therefore have a simple combinatorial condition for a code that guarantees it has no local obstructions:

\begin{corollary}\label{cor:max-int-complete}
If a neural code $\C$ is max $\cap$-complete, then $\C$ has no local obstructions.
\end{corollary}

For $n \leq 4$, the convex codes are precisely those codes that are max $\cap$-complete. 

\begin{proposition}\label{prop:n4}
Let $\C$ be a code on $n \leq 4$ neurons.  Then $\C$ is convex if and only if $\C$ is max $\cap$-complete.
\end{proposition}

This is shown in Supplementary Text S1, where we provide a complete classification of convex codes on $n=4$ neurons.
Proposition~\ref{prop:n4} does not extend to $n > 4,$ however, since beginning in $n=5$ there are convex codes that are not max $\cap$-complete (see Example~\ref{ex:ex1}).  This raises the question: are there max $\cap$-complete codes that are not convex?  In a previous version of this paper, we conjectured that all max $\cap$-complete codes are convex \cite{MRC}.  This conjecture has recently been proven \cite{intersection-complete}, using ideas similar to what we illustrate in Example~\ref{ex:2D}.

\begin{proposition}[{\cite[Theorem 4.4]{intersection-complete}}] \label{prop:max-int}
If $\C$ is max $\cap$-complete, then $\C$ is convex.
\end{proposition}

Finally, in Supplementary Text S3 we present some straightforward bounds on the minimal embedding dimension $d(\C)$, obtained using results about $d$-representability of the associated simplicial complex $\Delta(\C)$.  In particular, we find bounds from Helly's theorem and the Fractional Helly theorem.  
Unfortunately, these results all stem from $\Delta(\C)$.  In our classification of convex codes for $n \leq 4$, however, it is clear that the presence or absence of specific codewords can affect $d(\C)$, even if $\Delta(\C)$ remains unchanged (see Table~\ref{table:n4} in Supplementary Text S1).  The problem of how to use this additional information about a code in order to improve the bounds on $d(\C)$ remains wide open.

\section{Examples}\label{sec:examples}

Our first example depicts a convex code with minimal embedding dimension $d(\C) = 2$.

\begin{example}\label{ex:example1}
Consider the open cover $\U$ illustrated in Figure~\ref{fig:example1}a.  The corresponding code, $\C = \C(\U)$, has 10 codewords.  $\C$ is a convex code, by construction, and it is easy to check that $d(\C) = 2$.  The simplicial complex $\Delta(\C)$ (see Figure~\ref{fig:example1}b) loses some of the information about the cover $\U$ that is present in $\C$.  In particular, $U_2 \subseteq U_1 \cup U_3$ and $U_2 \cap U_4 \subseteq U_3$ is reflected in $\C$, but not in $\Delta(\C)$. Note that we can infer $U_2 \subseteq U_1 \cup U_3$ directly from the code, because any codeword with neuron $2$ ``on'' also has neuron $1$ or $3$ ``on.'' 
\end{example}

\vspace{-.1in}
\begin{figure}[!h]
\begin{center}
\includegraphics[width=4.5in]{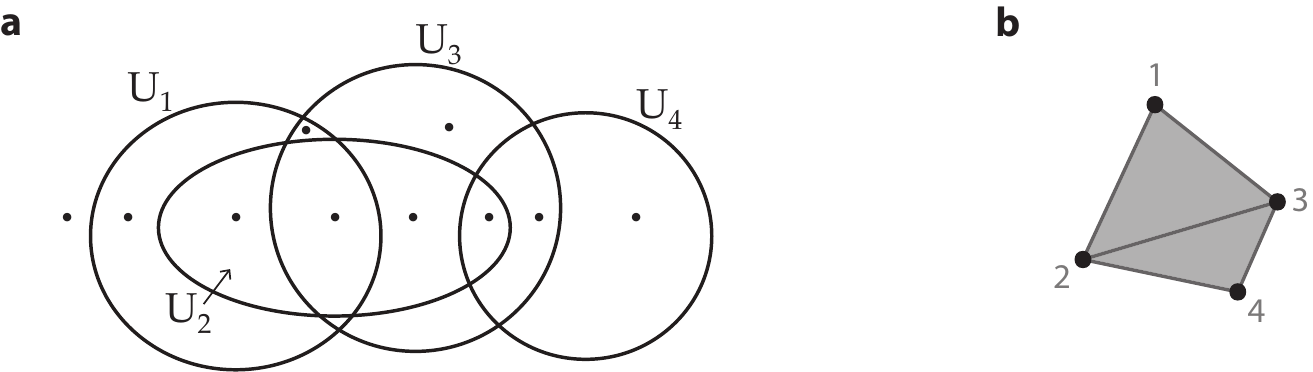}
\end{center}
\caption{{\bf (a)} An arrangement $\U = \{U_1, U_2, U_3, U_4\}$ of convex open sets.  Black dots mark regions corresponding to distinct codewords in $\C = \C(\U)$.  From left to right, the codewords are
0000, 1000, 1100, 1010, 1110, 0110, 0010, 0111, 0011, and 0001.   {\bf (b)} The simplicial complex $\Delta(\C)$ for the code $\C$ defined in (a).  The two facets, 123 and 234, correspond to the two maximal codewords, 1110 and 0111, respectively.  This simplicial complex is also equal to the nerve of the cover $\N(\U)$ (see Section~\ref{sec:nerve}).}
\label{fig:example1}
\end{figure}

Note that the convex code in Example~\ref{ex:example1} is also max $\cap$-complete, as guaranteed by Proposition~\ref{prop:n4}.  The next example shows that this proposition does not hold for $n \geq 5$. 

\begin{figure}[!h]
\begin{center}
\includegraphics[width=5in]{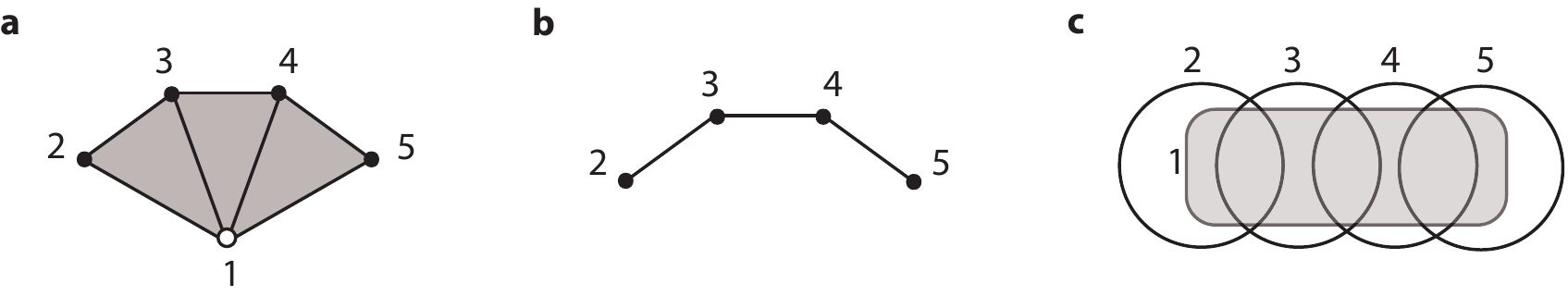}
\end{center}
\vspace{-.2in}
\caption{{\bf (a)} A simplicial complex $\Delta$ on $n=5$ vertices.  The vertex $1$ is an intersection of facets, but is not contained in the code $\C$ of Example~\ref{ex:ex1}. {\bf (b)} The link $\Lk_{1}(\Delta)$ (see Section~\ref{sec:local-obs}). {\bf (c)} A convex realization of the code $\C$.  The set $U_1$ corresponding to neuron $1$ (shaded) is completely covered by the other sets $U_2,\ldots,U_5$, consistent with the fact that $1 \notin \C$.}
\label{fig:fan}
\end{figure}

\begin{example}\label{ex:ex1}
The simplicial complex $\Delta$ shown in Figure~\ref{fig:fan}a has facets 123, 134, and 145.  Their intersections yield the faces 1, 13, and 14, so that $\mathcal{F}_\cap(\Delta) = \{123, 134, 145, 13, 14, 1,\emptyset\}.$
For this $\Delta$, we can compute the minimal code with no local obstructions, $\C_{\min}(\Delta) = \{123, 134, 145, 13, 14, \emptyset \},$ as in Theorem~\ref{thm:Thm1}.
Note that the element $1 \in \mathcal{F}_\cap(\Delta)$ is not present in $\C_{\min}(\Delta)$. 
Now consider the code $\C =  \Delta \setminus \{1\}$.  Clearly, this code has simplicial complex $\Delta(\C) = \Delta\,$: it has a codeword for each face of $\Delta$, except the vertex $1$ (see Figure~\ref{fig:fan}a).
By Theorem~\ref{thm:Thm1}, $\C$ has no local obstructions because $\C \supseteq \C_{\min}(\Delta)$.  However, 
$\C$ is not
max $\cap$-complete because $\mathcal{F}_\cap(\Delta) \not\subseteq \C$.  Nevertheless, $\C$ is convex.  A convex realization is shown in  Figure~\ref{fig:fan}c.
\end{example}

The absence of local obstructions is a necessary condition for convexity.  Unfortunately, it is not sufficient: the following example shows a code with no local obstructions that is {\it not} convex.

\begin{example}[\cite{counterexample}]\label{ex:counterexample}
Consider the code $\C = \{2345, 123, 134, 145, 13, 14, 23, 34, 45, 3, 4, \emptyset\}$.  The simplicial complex of this code, $\Delta=\Delta(\C),$ has facets $\{2345, 123, 134, 145\}$.
It is straightforward to show that $\C = \C_{\min}(\Delta)$, and thus $\C$ has no local obstructions.  Despite this, it was shown in \cite{counterexample} using geometric arguments that $\C$ is {\it not} convex.  Note that this code is not max $\cap$-complete 
(the max intersection $123 \cap 134 \cap 145 = 1$ is not in $\C$).
\end{example}

The next example illustrates how a code with a single maximal codeword can be embedded in $\RR^2$.  This basic construction is used repeatedly in our proof of Proposition~\ref{prop:n4} (see Supplementary Text S1), and inspired aspects of the proof of Proposition~\ref{prop:max-int}, given in \cite{intersection-complete}.
\bigskip

\begin{figure}[h!]
\begin{center}
\includegraphics[width=4in]{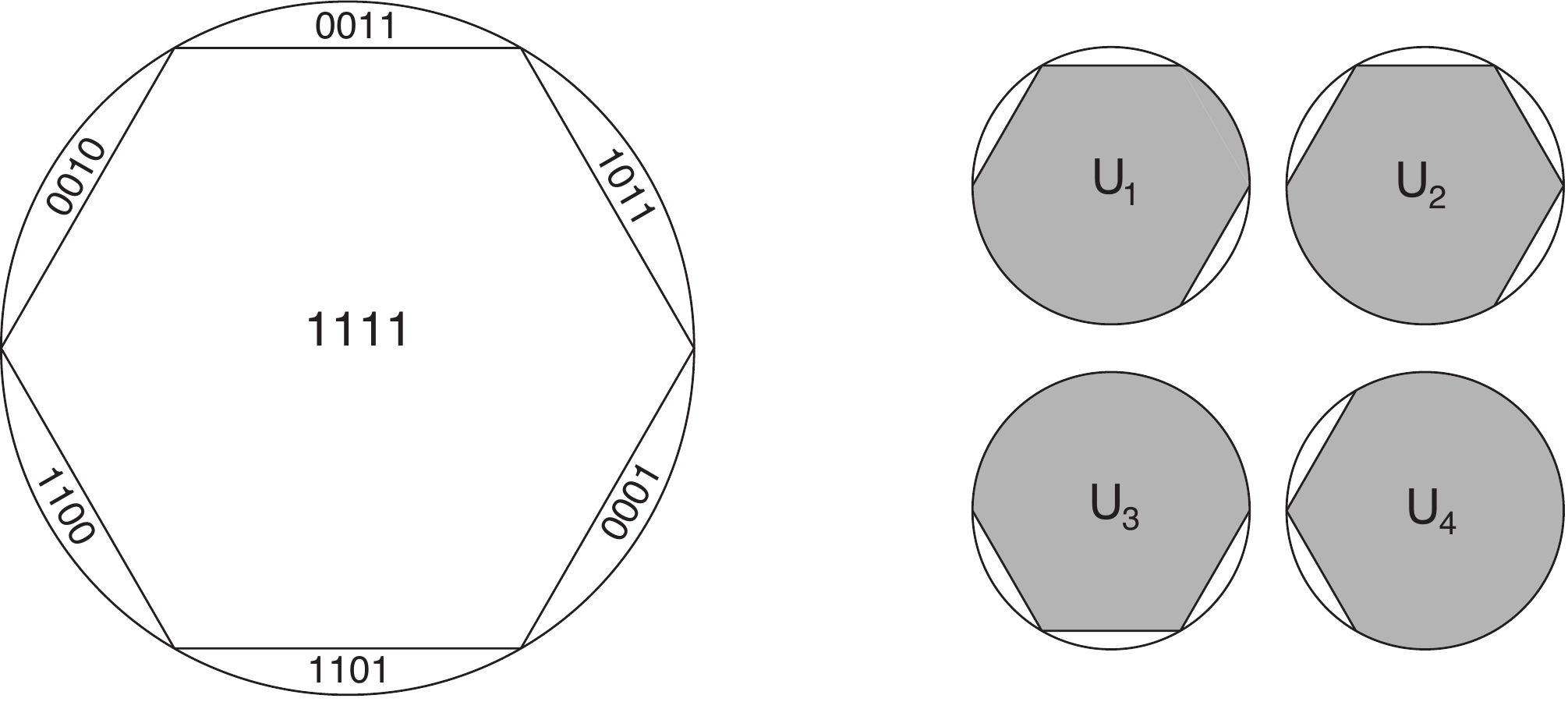}
\end{center}
\caption{A convex realization in $\mathbb{R}^2$ of the code in Example~\ref{ex:2D}.  (Left) Each non-maximal codeword is assigned a region outside the polygon, but inside the disk.  (Right) For each neuron $i$, the convex set $U_i$ is the union of all regions corresponding to codewords with a $1$ in the $i^\mathrm{th}$ position.}
\label{fig:ngon}
\end{figure}

\begin{example}\label{ex:2D} 
Consider the code $\C = \{1111,1011,1101,1100,0011,0010,0001, 0000 \}$, with unique maximal codeword 1111. Figure~\ref{fig:ngon} depicts 
the construction of a convex realization in $\mathbb{R}^2$.  All regions corresponding to codewords are subsets of a disk in $\RR^2$.  For each $i=1,\ldots,4$, the convex set $U_i$ is the union of all regions whose corresponding codewords have a $1$ in the $i^\mathrm{th}$ position.  For example, $U_1$ is the union of the four regions corresponding to codewords 1111, 1011, 1101, and 1100.
\end{example}

The above construction can be generalized to any code with a unique maximal codeword.

\begin{lemma}\label{lemma:D2}
Let $\C$ be a code with a unique maximal codeword.  Then $\C$ is convex, and $d(\C) \leq 2$.
\end{lemma}

\begin{proof}
Let $\rho \in \C$ be the unique maximal codeword, and let $m = |\C|-2$ be the number of non-maximal codewords, excluding the all-zeros word.  Inscribe a regular open $m$-gon $P$ in an open disk, so that there are $m$ sectors surrounding $P$, as in Figure~\ref{fig:ngon}.  (If $m<3$, let $P$ be an open triangle.)  Assign each non-maximal codeword (excluding $00\cdots0$) to a distinct sector inside the disk but outside of $P$, and assign the maximal codeword $\rho$ to $P$.  Next, for each $i \in \rho$ let $U_i$ be the union of $P$ and all sectors whose corresponding codewords have a $1$ in the $i^\mathrm{th}$ position, together with their common boundaries with $P$.  For $j \in [n] \setminus \rho$, set $U_j = \emptyset$.  Note that each $U_i$ is open and convex, and $\C = \C(\{U_i\}).$
\end{proof}

Lemma~\ref{lemma:D2} can easily be generalized to any code whose maximal codewords are non-overlapping (that is, having disjoint supports).  In this case, each nonzero codeword is contained in a unique facet of $\Delta(\C)$, and the facets thus yield a partition of the code.  We can repeat the above construction in parallel for each part, obtaining the same dimension bound.

\begin{proposition}\label{prop:disjointsimplices}
Let $\C$ be a code with non-overlapping maximal codewords (i.e., the facets of $\Delta(\C)$ are disjoint).  Then $\C$ is convex and $d(\C)\leq 2$. 
\end{proposition}

\section{Local obstructions to convexity}\label{sec:local-obstructions}

For any simplicial complex $\Delta$, there exists a convex cover $\U$ in a high-enough dimensional space $\RR^d$ such that $\Delta$ can be realized as $\Delta(\C(\U))$ \cite{tancer-survey}.  For this reason, the simplicial complex $\Delta(\C)$ alone is not sufficient to determine whether or not $\C$ is convex.
Obstructions to convexity must emerge from information in the code that goes beyond what is reflected in $\Delta(\C)$.  As was shown in \cite{neural_ring}, this additional information is precisely the receptive field relationships, which we turn to now.

\subsection{Receptive field relationships}\label{sec:RF}

For a code $\C$ on $n$ neurons, let $\U = \{U_1, \ldots, U_n\}$ be any collection of open sets such that $\C = \C(\U)$, and recall that $U_\sigma = \bigcap_{i \in \sigma} U_i$.

\begin{definition}
A {\em receptive field relationship}  (RF relationship) of $\C$ is a pair $(\sigma,\tau)$ corresponding to the set containment
$$U_\sigma \subseteq \bigcup_{i \in \tau} U_i,$$
 where $\sigma \neq \emptyset$, $\sigma \cap \tau = \emptyset$, and $U_\sigma \cap U_i \neq \emptyset$ for all $i \in \tau$.   
 \end{definition}
If $\tau = \emptyset$, then the relationship $(\sigma,\emptyset)$ simply states that $U_\sigma = \emptyset.$  Note that relationships of the form $(\sigma,\emptyset)$ reproduce the information in $\Delta(\C)$, while those of the form $(\sigma,\tau)$ for $\tau \neq \emptyset$ reflect additional structure in $\C$ that goes beyond the simplicial complex.
 A {\it minimal} RF relationship is one such that no single neuron can be removed from $\sigma$ or $\tau$ without destroying the containment. 
 
 It is important to note that RF relationships are independent of the choice of open sets $\U$ (see Lemma 4.2 of \cite{neural_ring}).  Hence we denote the set of all RF relationships $\{(\sigma,\tau)\}$ for a given code 
 $\C$ as simply $\RF(\C).$ 
 In \cite{neural_ring}, it was shown that one can compute $\RF(\C)$ algebraically, using an associated ideal $I_\C$.

\begin{example}[Example~\ref{ex:example1} continued]\label{ex:RF}
The code $\C=\C(\U)$ from Example~\ref{ex:example1} has RF relationships 
$\RF(\C) = \left\{(\{1,4\},\emptyset), (\{1,2,4\},\emptyset), (\{1,3,4\},\emptyset),(\{1,2,3,4\},\emptyset),(\{2\}, \{1,3\}),(\{2\}, \{1,3, 4\}),(\{2, 4\}, \{3\})\right\}$.  Of these, the pairs $(\{1,4\},\emptyset)$, $(\{2\}, \{1,3\})$, and $(\{2, 4\}, \{3\}),$ corresponding to $U_1 \cap U_4 = \emptyset$, $U_2 \subseteq U_1 \cup U_3$, and $U_2 \cap U_4 \subseteq U_3$ respectively, are the minimal RF relationships.  
\end{example}

The following lemma illustrates a simple case where RF relationships can be used to show that a code cannot have a convex realization.  (This is a special case of Lemma~\ref{lemma:loc-obs} below.)

\begin{lemma}  \label{lemma:disconnect}
Let $\C = \C(\U)$.  If $\C$ has RF relationships $U_\sigma \subseteq U_i \cup U_j$  and $U_\sigma \cap U_i \cap U_j = \emptyset$ for some $\sigma \subseteq [n]$ and distinct $i, j \notin \sigma$, then $\C$ is \underline{not} a convex code.
\end{lemma}

\begin{proof}  By assumption, $\left\{(\sigma, \{i,j\}), (\sigma\cup\{i,j\},\emptyset)\right\} \subseteq \RF(\C).$
It follows that the sets $V_i = U_\sigma \cap U_i \neq \emptyset$ and $V_j = U_\sigma \cap U_j \neq \emptyset$ are disjoint open sets that each intersect $U_\sigma$, and $U_\sigma \subseteq V_i \cup V_j$.  We can thus conclude that $U_\sigma$ is disconnected in any open cover $\U$ such that $\C = \C(\U)$.  This implies that $\C$ cannot have a convex realization, because if the $U_i$s were all convex then $U_\sigma$ would be convex, contradicting the fact that it is disconnected.
\end{proof}

The above proof relies on the observation that $U_\sigma$ must be convex in any convex realization $\U$, but the properties of the code imply that $U_\sigma$ is covered by a collection of open sets whose topology does not match that of a convex set.  This topological inconsistency between a set and its cover is, at its core, a contradiction arising from the Nerve lemma, which we discuss next.

\subsection{The Nerve lemma}\label{sec:nerve}
The {\it nerve} of an open cover $\U = \{U_1,\ldots,U_n\}$ is the simplicial complex
$$\N(\U) \od \{\sigma \subseteq [n] \mid U_\sigma \neq \emptyset \}.$$
In fact, $\N(\U) = \Delta(\C(\U)),$ so the nerve can be recovered directly from the code $\C(\U)$.
The Nerve lemma tells us that $\N(\U)$ carries a surprising amount of topological information about the underlying space covered by $\U$, provided $\U$ is a good cover.  Recall that a {\it good cover} is a collection of open sets $\{U_i\}$ where
 every non-empty intersection, $U_\sigma = \bigcap_{i \in \sigma} U_i,$
is contractible.\footnote{A set is {\it contractible} if it is homotopy-equivalent to a point \cite{Hatcher}.} 

\begin{lemma}[Nerve lemma]  If $\U$ is a good cover, then $\bigcup_{i=1}^n U_i$ is homotopy-equivalent to $\N(\U)$.  In particular, $\bigcup_{i=1}^n U_i$ and $\N(\U)$ have exactly the same homology groups.
\end{lemma}

\noindent This result is well known, and can be obtained as a direct consequence of \cite[Corollary 4G.3]{Hatcher}.
 
Now observe that an open cover comprised of convex sets is always a good cover, because the intersection of convex sets is convex, and hence contractible.
For example, if $\C=\C(\U)$ for a convex cover $\U$, then $\Delta(\C)$ must match the homotopy type of $\bigcup_{i=1}^n U_i$.  This fact was previously exploited to extract topological information about the represented space from hippocampal place cell activity \cite{cellgroups}.

The Nerve lemma is also key to our notion of local obstructions, which we turn to next.

\subsection{Local obstructions}\label{sec:local-obs}

Local obstructions arise when a code contains a RF relationship $(\sigma, \tau)$, so that $U_\sigma \subseteq \bigcup_{i\in\tau} U_i,$ but the nerve of the corresponding cover of $U_\sigma$ by the restricted sets $\{U_i \cap U_\sigma\}_{i \in \tau}$ is not contractible.  By the Nerve lemma, if the $U_i$s are all convex then $\N(\{U_\sigma \cap U_i\}_{i \in \tau})$ must have the same homotopy type as $U_\sigma$, which is contractible.  If $\N(\{U_\sigma \cap U_i\}_{i \in \tau})$ fails to be contractible, we can conclude that the $U_i$s can not all be convex.

Now, observe that the nerve of the restricted cover $\N(\{U_\sigma \cap U_i\}_{i \in \tau})$ is related to the nerve of the original cover $\N(\U)$ as follows:
$$\N(\{U_\sigma \cap U_i\}_{i \in \tau}) = \{\omega \in \N(\U) \mid \sigma \cap \omega = \emptyset,\;
\sigma \cup \omega \in \N(\U), \text{ and } \omega \subseteq \tau\}.$$
In fact, letting $\Delta = \N(\U)$ and considering the restricted complex $\Delta|_{\sigma \cup \tau}$, we recognize that the right-hand side above is precisely the link,
$$\N(\{U_\sigma \cap U_i\}_{i \in \tau}) = \Lk_\sigma(\Delta|_{\sigma \cup \tau}).$$
We can now define a {\it local obstruction} to convexity.  

\begin{definition} Let $(\sigma,\tau) \in \RF(\C)$, and let $\Delta = \Delta(\C)$.
We say that $(\sigma,\tau)$ is a {\em local obstruction} of $\C$ if $\tau \neq \emptyset$ and $\Lk_\sigma(\Delta|_{\sigma \cup \tau})$ is {\it not} contractible.
\end{definition}

Local obstructions are thus detected via non-contractible links of the form $\Lk_\sigma(\Delta|_{\sigma \cup \tau})$, where $(\sigma,\tau) \in \RF(\C)$.  Figure~\ref{fig:all-links} displays all possible links that can arise for $|\tau| \leq 4$.  Non-contractible links are highlighted in red.  
Note also that $\tau \neq \emptyset$ implies $\sigma \notin \C$ and $U_\sigma \neq \emptyset,$ as the definition of a RF relationship requires that $U_\sigma \cap U_i \neq \emptyset$ for all $i \in \tau$.  Any local obstruction $(\sigma,\tau)$ must therefore have $\sigma \in \Delta(\C) \setminus \C$ and $\Lk_\sigma(\Delta|_{\sigma \cup \tau})$ nonempty.  

The arguments leading up to the definition of local obstruction imply the following simple consequence of the Nerve lemma, which was previously observed in \cite{no-go}.

\begin{figure}
\begin{center}
\includegraphics[width = 5.25in]{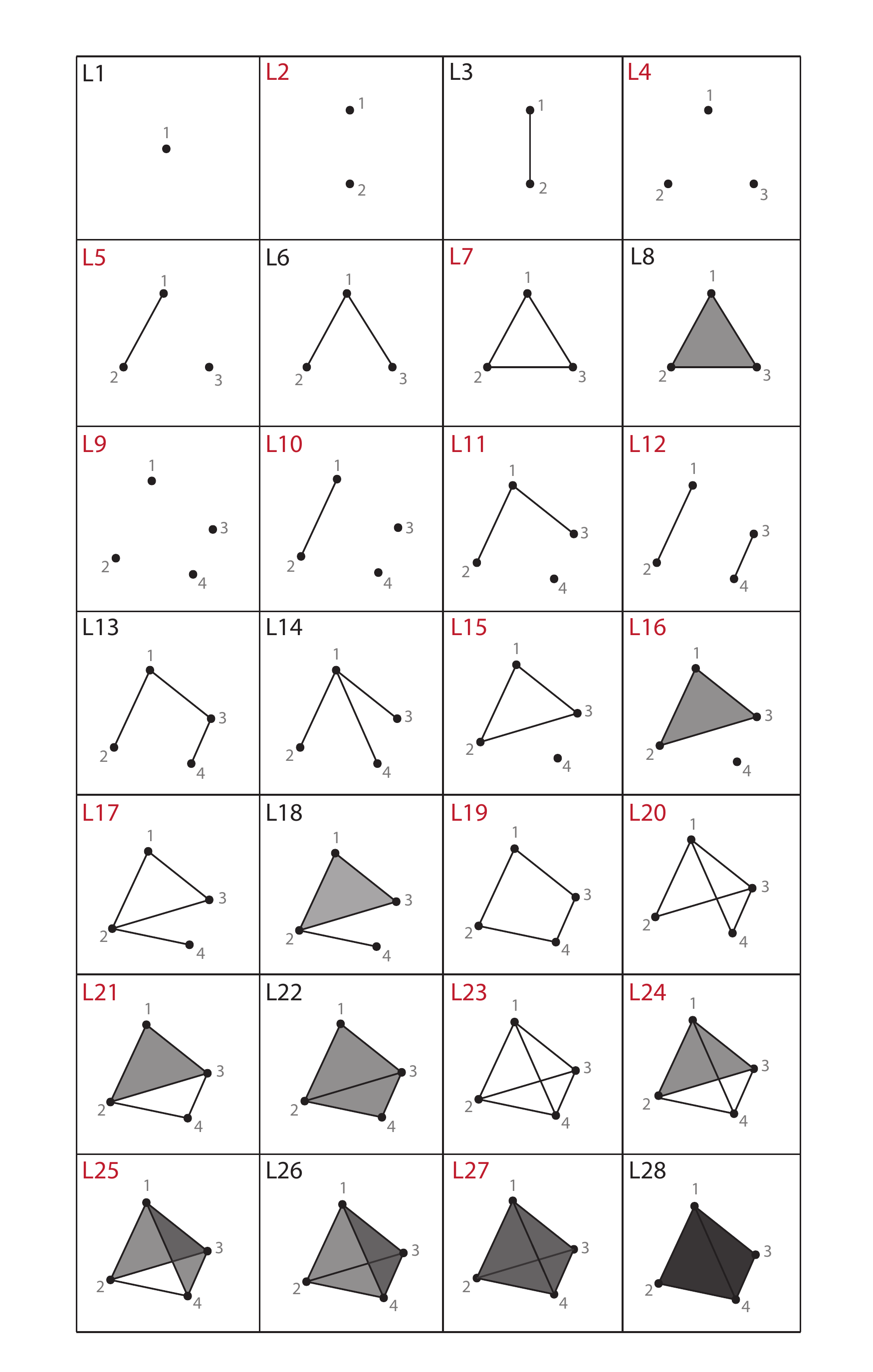}
\end{center}
\vspace{-.4in}
\caption{\small All simplicial complexes on up to $n = 4$ vertices, up to permutation equivalence.  These can each arise as links of the form $\Lk_{\sigma}(\Delta|_{\sigma \cup \tau})$ for $|\tau| \leq 4$.  Red labels correspond to non-contractible complexes.  Note that L13 is the only simplicial complex on $n\leq 4$ vertices that is contractible but {\it not} a cone.} 
\label{fig:all-links}
\end{figure}

\begin{lemma}[Lemma~\ref{lemma1}] \label{lemma:loc-obs}
If $\C$ has a local obstruction, then $\C$ is not a convex code.
\end{lemma}

In general, the question of whether or not a given simplicial complex is contractible is undecidable \cite{Tancer-undecidable}; however, in some cases it is easy to see that all relevant links will be contractible.  This yields a simple condition on RF relationships that guarantees that a code has no local obstructions.

\begin{lemma}\label{lemma:loc-good}
Let $\C = \C(\U)$.  If for each $(\sigma,\tau) \in \RF(\C)$ we have $U_{\sigma} \cap U_{\tau} \neq \emptyset$, then $\C$ has no local obstructions.
\end{lemma}

\begin{proof}
Let $\Delta = \Delta(\C)$.
$U_\sigma \cap U_\tau \neq \emptyset$ implies $\Lk_\sigma(\Delta|_{\sigma \cup \tau})$ is the full simplex on the vertex set $\tau$, which is contractible.  If this is true for every RF relationship, then none can give rise to a local obstruction.
\end{proof}

\noindent For example, if $11 \cdots 1 \in \C$, then $U_\sigma \cap U_\tau \neq \emptyset$ for any pair $\sigma,\tau \subset [n]$, so $\C$ has no local obstructions.

\section{Characterizing local obstructions via mandatory codewords}\label{sec:mandatory-codewords}

From the definition of local obstruction, it seems that in order to show that a code has no local obstructions one would need to check the contractibility of {\it all} links of the form $\Lk_\sigma(\Delta|_{\sigma \cup \tau})$ corresponding to all pairs $(\sigma,\tau) \in \RF(\C)$.   We shall see in this section that in fact we only need to check for contractibility of links inside the full complex $\Delta$ -- that is, links of the form $\Lk_\sigma(\Delta)$.  This is key to obtaining a list of mandatory codewords, $\C_{\min}(\Delta)$, that depends only on $\Delta$, and not on any further details of the code.

In Section~\ref{sec:link-lemmas} we prove some important lemmas about links, and then use them in Section~\ref{sec:Thm1} to prove Theorem~\ref{thm:Thm1}.

\subsection{Link lemmas} \label{sec:link-lemmas}
 
In what follows, the notation
$$\operatorname{cone}_v(\Delta) \od \left\{\{v\} \cup \omega \mid \omega \in \Delta\right\} \cup \Delta$$
denotes the {\it cone} of $v$ over $\Delta$, where $v$ is a new vertex not contained in $\Delta$.  Any simplicial complex that is a cone over a sub-complex, so that $\Delta = \operatorname{cone}_v(\Delta')$, is automatically contractible.  In Figure~\ref{fig:all-links}, the only contractible link that is {\it not} a cone is L13.   This is the same link that appeared in Figure~\ref{fig:fan}b of Example~\ref{ex:ex1}.

\begin{lemma}\label{lemma:link1}
Let $\Delta$ be a simplicial complex on $[n]$, $\sigma \in \Delta$, and $v \in [n]$ such that $v \notin \sigma$ and $\sigma \cup \{v\} \in \Delta$.  Then $\Lk_{\sigma \cup \{v\}}(\Delta) \subseteq \Lk_\sigma(\Delta|_{[n]\setminus\{v\}}),$ and
$$\Lk_\sigma(\Delta) = \Lk_\sigma(\Delta|_{[n]\setminus\{v\}}) \cup \operatorname{cone}_v(\Lk_{\sigma \cup \{v\}}(\Delta)).$$
\end{lemma}

\begin{proof}
The proof follows from the definition of the link.  First, observe that
\begin{eqnarray*}
\Lk_{\sigma \cup \{v\}}(\Delta) &=& \{\omega \subset [n] \mid v \notin \omega, \omega \cap \sigma = \emptyset, \text{ and } \omega \cup \sigma \cup \{v\} \in \Delta\}\\
&=& \{\omega \subset [n]\setminus\{v\} \mid \omega \cap \sigma = \emptyset, \text{ and } \omega \cup \sigma \cup \{v\} \in \Delta\}\\
& \subseteq &  \{\omega \subset [n]\setminus\{v\} \mid \omega \cap \sigma = \emptyset, \text{ and } \omega \cup \sigma \in \Delta|_{[n]\setminus\{v\}}\}\\
&=& \Lk_\sigma(\Delta|_{[n]\setminus\{v\}}),
\end{eqnarray*}
which establishes that $\Lk_{\sigma \cup \{v\}}(\Delta) \subseteq \Lk_\sigma(\Delta|_{[n]\setminus\{v\}}).$  Next, observe that
\begin{eqnarray*}
\operatorname{cone}_v(\Lk_{\sigma \cup \{v\}}(\Delta)) \setminus \Lk_{\sigma \cup \{v\}}(\Delta)  &=& \{\{v\} \cup \omega \mid \omega \in \Lk_{\sigma \cup \{v\}}(\Delta)\}\\
&=& \{\{v\} \cup \omega \mid \omega \subset [n], v \notin \omega, \omega \cap \sigma = \emptyset, \text{ and } \omega \cup \{v\} \cup \sigma \in \Delta\}\\
&=& \{\tau \subset [n] \mid v \in \tau, \tau \cap \sigma = \emptyset, \text{ and } \tau \cup \sigma \in \Delta\}\\
&=& \{\omega \in \Lk_\sigma(\Delta) \mid v \in \omega\}.
\end{eqnarray*}
Finally,
$$\Lk_\sigma(\Delta|_{[n]\setminus\{v\}}) = \{\omega \in \Lk_\sigma(\Delta) \mid v \notin \omega\}.$$
From here the second statement is clear.
\end{proof}

\begin{corollary}
Let $\Delta$ be a simplicial complex on $[n]$, $\sigma \in \Delta$, and $v \in [n]$ such that $v \notin \sigma$ and 
$\sigma \cup \{v\} \in \Delta$. 
If $\Lk_{\sigma \cup \{v\}}(\Delta)$ is contractible, then $\Lk_\sigma(\Delta)$ and $\Lk_\sigma(\Delta|_{[n]\setminus\{v\}})$ are homotopy-equivalent.
\end{corollary}

\begin{proof}
Lemma~\ref{lemma:link1} states that $\Lk_{\sigma \cup \{v\}}(\Delta) \subseteq \Lk_\sigma(\Delta|_{[n]\setminus\{v\}}),$ and that $\Lk_\sigma(\Delta)$ can be obtained from $\Lk_\sigma(\Delta|_{[n]\setminus\{v\}})$ by coning off the subcomplex $\Lk_{\sigma \cup \{v\}}(\Delta)$ -- that is, by including $\operatorname{cone}_v(\Lk_{\sigma \cup \{v\}}(\Delta))$.  If this subcomplex is itself contractible, then the homotopy type is preserved.
\end{proof}

\noindent Another useful corollary follows from the one above by simply setting $\Delta = \Delta|_{\sigma \cup \tau \cup \{v\}}$ and $[n] = \sigma \cup \tau \cup \{v\}.$  We immediately see that if both $\Lk_\sigma(\Delta|_{\sigma \cup \tau \cup \{v\}})$ and $\Lk_{\sigma \cup \{v\}}(\Delta|_{\sigma \cup \tau \cup \{v\}})$ are contractible, then $\Lk_\sigma(\Delta|_{\sigma \cup \tau})$ is contractible.  Stated another way:

\begin{corollary}\label{cor:biglink}     Assume $v \notin \sigma$ and $\sigma \cap \tau = \emptyset$. If $\Lk_\sigma(\Delta|_{\sigma \cup \tau})$ is {\it not} contractible, then either (i) $\Lk_\sigma(\Delta|_{\sigma \cup \tau \cup \{v\}})$ is not contractible, and/or (ii) $\sigma \cup \{v\} \in \Delta$ and $\Lk_{\sigma \cup \{v\}}(\Delta|_{\sigma \cup \tau \cup \{v\}})$ is not contractible.
\end{corollary}

This corollary can be extended to show that for every non-contractible link $\Lk_\sigma(\Delta|_{\sigma \cup \tau})$, there exists a non-contractible ``big'' link $\Lk_{\sigma'}(\Delta)$  for some $\sigma' \supseteq \sigma.$ This is because vertices outside of $\sigma \cup \tau$ can be added one by one to either $\sigma$ or its complement, preserving the non-contractibility of the new link at each step.  (Note that if $\sigma \cup \{v\} \notin \Delta$, we can always add $v$ to the complement.  In this case, $\Lk_\sigma(\Delta|_{\sigma \cup \tau}) = \Lk_{\sigma}(\Delta|_{\sigma \cup \tau \cup \{v\}})$, so we are in case (i) of Corollary~\ref{cor:biglink}.)  In other words, we have the following lemma:

\begin{lemma}\label{lemma:bigbadlink}
Let $\sigma, \tau \in \Delta$.
Suppose $\sigma \cap \tau = \emptyset$, and $\Lk_\sigma(\Delta|_{\sigma \cup \tau})$ is {\it not} contractible.  Then there exists $\sigma' \in \Delta$ such that $\sigma' \supseteq \sigma$, $\sigma' \cap \tau = \emptyset$, and $\Lk_{\sigma'}(\Delta)$ is not contractible.
\end{lemma}

The next results show that only intersections of facets (maximal faces under inclusion) can possibly yield non-contractible links.  For any $\sigma \in \Delta$, we denote by $f_\sigma$ the intersection of all facets of $\Delta$ containing $\sigma$.
In particular, $\sigma = f_\sigma$ if and only if $\sigma$ is an intersection of facets of $\Delta$.  It is also useful to observe that a simplicial complex is a cone if and only if the common intersection of all its facets is non-empty.  (Any element of that intersection can serve as a cone point, and a cone point is necessarily contained in all facets.)

\begin{lemma}
Let $\sigma \in \Delta$.  Then $\sigma = f_\sigma \; \Leftrightarrow \; \Lk_\sigma(\Delta)$ is {\it not} a cone.
\end{lemma}

\begin{proof}
Recall that $\Lk_\sigma(\Delta)$ is a cone if and only if all facets of  $\Lk_\sigma(\Delta)$ have a non-empty common intersection $\nu$.   This can happen if and only if $\sigma \cup \nu \subseteq f_\sigma$.  Note that since $\nu \in \Lk_\sigma(\Delta)$, we must have $\nu \cap \sigma = \emptyset$ and hence $\Lk_\sigma(\Delta)$ is a cone if and only if $\sigma \neq f_\sigma$.
\end{proof}

Furthermore, it is easy to see that every simplicial complex that is not a cone can in fact arise as the link of 
an intersection of facets.  For any $\Delta$ that is not a cone, simply consider $\widetilde{\Delta} = \operatorname{cone}_v(\Delta)$; $v$ is an intersection of facets of $\widetilde{\Delta}$, and $\Lk_{v}(\widetilde{\Delta}) = \Delta.$  

The above lemma immediately implies the following corollary:

\begin{corollary}\label{cor:contractible-link}
Let $\sigma \in \Delta$ be nonempty.  If $\sigma \neq f_\sigma$, then $\Lk_\sigma(\Delta)$ is a cone and hence contractible.  In particular, if $\Lk_\sigma(\Delta)$ is {\it not} contractible, then $\sigma$ must be an intersection of facets of $\Delta$ (i.e., $\sigma \in \mathcal{F}_{\cap}(\Delta)$).
\end{corollary}

Finally, we note that all pairwise intersections of facets that are not also higher-order intersections give rise to non-contractible links.

\begin{lemma}\label{lemma:pairwise}
Let $\Delta$ be a simplicial complex.  If $\sigma = \tau_1 \cap \tau_2,$ where $\tau_1, \tau_2$ are distinct facets of $\Delta$, and $\sigma$ is not contained in any other facet of $\Delta$, then $\Lk_{\sigma}(\Delta)$ is not contractible.
\end{lemma}

\begin{proof}
Observe that $\Lk_\sigma(\Delta)$ consists of all subsets of $\omega_1 = \tau_1 \setminus \sigma$ and $\omega_2 = \tau_2 \setminus \sigma$, but $\omega_1$ and $\omega_2$ are disjoint because $\tau_1$ and $\tau_2$ do not overlap outside of $\sigma$.  This means $\Lk_\sigma(\Delta)$ has two connected components, and is thus not contractible.
\end{proof}

Note that if $\sigma$ is a pairwise intersection of facets that is also contained in another facet, then $\Lk_{\sigma}(\Delta)$ could be contractible.  For example, the vertex $1$ in Figure~\ref{fig:fan}a can be expressed as a pairwise intersection of facets $123$ and $145$, but is also contained in $134$.  As shown in Figure~\ref{fig:fan}b, the corresponding link $\Lk_1(\Delta)$ is contractible.

\subsection{Proof of Theorem~\ref{thm:Thm1} and Lemma~\ref{lemma:key-lemma}}\label{sec:Thm1}
Using the above facts about links, we can now prove Theorem~\ref{thm:Thm1} and Lemma~\ref{lemma:key-lemma}.  First, we need the following key proposition.

\begin{proposition}\label{prop:big-thm}
A code $\C$ has no local obstructions if and only if $\sigma \in \C$ for every $\sigma \in \Delta(\C)$ such that
$\Lk_\sigma(\Delta(\C))$ is non-contractible.
\end{proposition}

\begin{proof} Let $\Delta = \Delta(\C)$, and let $\U = \{U_i\}$ be any collection of open sets such that $\C = \C(\U)$.
($\Rightarrow$)  We prove the contrapositive.  Suppose there exists $\sigma \in \Delta(\C) \setminus \C$ such that $\Lk_\sigma(\Delta)$ is non-contractible.  Then $U_\sigma$ must be covered by the other sets $\{U_i\}_{i \notin \sigma}$, and since $\Lk_\sigma(\Delta)$ is not contractible, the RF relationship $(\sigma, \bar\sigma)$ is a local obstruction. 
($\Leftarrow$)  We again prove the contrapositive.
Suppose $\C$ has a local obstruction $(\sigma,\tau)$.  This means that $\sigma \cap \tau = \emptyset$,  $U_{\sigma} \subseteq \bigcup_{i\in \tau} U_i$, and $\Lk_{\sigma}(\Delta|_{\sigma \cup \tau})$ is not contractible.  By Lemma~\ref{lemma:bigbadlink}, there exists $\sigma' \supseteq \sigma$ such that $\sigma' \cap \tau = \emptyset$ and $\Lk_{\sigma'}(\Delta)$ is not contractible.  Moreover, $U_{\sigma'} \subseteq U_{\sigma} \subseteq \bigcup_{i\in \tau} U_i$ with $\sigma' \cap \tau = \emptyset$, which implies $\sigma' \notin \C$. 
\end{proof}

Theorem~\ref{thm:Thm1} now follows as a corollary of Proposition~\ref{prop:big-thm}.  To see this, let
$$\C_{\min}(\Delta) = \{\sigma \in \Delta \mid \Lk_\sigma(\Delta) \text{ is non-contractible}\} \cup \{\emptyset\},$$
and note that $\C_{\min}(\Delta)$ has simplicial complex $\Delta$. This is because for any facet $\rho \in \Delta$,   $\Lk_\rho(\Delta) = \emptyset,$ which is non-contractible, and thus $\C_{\min}(\Delta)$ contains all the facets of $\Delta$.
By Proposition~\ref{prop:big-thm}, any code $\C$ with simplicial complex $\Delta$ has no local obstructions precisely when $\C \supseteq \C_{\min}(\Delta)$.  Thus, $\C_{\min}(\Delta)$ is the unique code satisfying the required properties in Theorem~\ref{thm:Thm1}.

Finally, it is easy to see that Lemma~\ref{lemma:key-lemma} follows directly from Corollary~\ref{cor:contractible-link}.

\section{Computing mandatory codewords algebraically}\label{sec:hochster}

Computing $\C_{\min}(\Delta)$ is certainly simpler than finding {\it all} local obstructions.  However, it is still difficult in general because determining whether or not a simplicial complex is contractible is undecidable \cite{Tancer-undecidable}.  For this reason, we now consider the subset of $\C_{\min}(\Delta)$ corresponding to non-contractible links that can be detected via homology:
\begin{equation}\label{eq:M_H}
\M_H(\Delta) \od \{\sigma \in \Delta \mid \dim \widetilde{H}_i(\Lk_\sigma(\Delta), \kk) > 0 \text{ for some } i\},
\end{equation}
where the $\widetilde{H}_i(\cdot)$  are reduced simplicial homology groups, and $\kk$ is a field.  Homology groups are topological invariants that can be easily computed for any simplicial complex, and {\it reduced} homology groups simply add a shift in the dimension of $\widetilde{H}_0(\cdot)$.  This shift is designed so that for any contractible space $X$, $\dim \widetilde{H}_i(X, \kk) = 0$ for all integers $i$.  
Clearly, $\M_H(\Delta) \subseteq \C_{\min}(\Delta),$ and $\M_H(\Delta)$ is thus a subset of the mandatory codewords that must be included in any convex code $\C$ with $\Delta(\C) = \Delta$.\footnote{Note that while $\M_H(\Delta)$ depends on the choice of field $\kk$, $\M_H(\Delta) \subseteq \C_{\min}(\Delta)$ for any $\kk$.}  On the other hand, $\M_H(\Delta) \subseteq \C$ does not guarantee that $\C$ has no local obstructions, as a homologically trivial simplicial complex may be non-contractible.\footnote{For example, consider a triangulation of the punctured Poincar\'e homology sphere: this simplicial complex has all-vanishing reduced homology groups, but is non-contractible \cite{homology-sphere}.} 

It turns out that the entire set $\M_{H}(\Delta)$ can be computed algebraically, via a minimal free resolution of an ideal built from $\Delta$.  Specifically,

\begin{equation}\label{eq:M_H-betti}
 \M_{H}(\Delta) = \{\sigma \in \Delta \mid \beta_{i,\bar\sigma}(S/I_{\Delta^*}) > 0 \text{ for some } i>0\},
 \end{equation}
where $S = \kk[x_1,\ldots,x_n]$, the ideal $I_{\Delta^*}$ is the Stanley-Reisner ideal of the Alexander dual $\Delta^*$, and $\beta_{i,\bar\sigma}(S/I_{\Delta^*})$ are the Betti numbers  of a minimal free resolution of the ring $S/I_{\Delta^*}$.
This is a direct consequence of Hochster's formula \cite{miller-sturmfels}:
\begin{equation}
\dim \widetilde{H}_i(\Lk_\sigma(\Delta), \kk) = \beta_{i+2,\bar\sigma}(S/I_{\Delta^*}).
\end{equation}
See Supplementary Text S2 for more details on Alexander duality, Hochster's formula, and the Stanley-Reisner ideal.

Moreover, the subset of mandatory codewords $\M_{H}(\Delta)$ can be easily computed using existing computational algebra software, such as {\tt Macaulay2} \cite{M2}.  We now describe this via an explicit example.  

\begin{example}\label{ex:L25}
Let $\Delta$ be the simplicial complex L25 in Figure~\ref{fig:all-links}.  The Stanley-Reisner ideal is given by $I_\Delta = \langle x_1x_2x_4, x_2x_3x_4 \rangle$, and its Alexander dual is 
$I_{\Delta^*} = \langle x_1,x_2,x_4\rangle \cap \langle x_2,x_3,x_4 \rangle = \langle x_1x_3,x_2,x_4\rangle.$
A minimal free resolution of $S/I_{\Delta^*}$ is:
\begin{eqnarray*}
 0 \longleftarrow S/I_{\Delta^*} \xleftarrow{[\begin{array}{ccc}x_1x_3 & x_2 & x_4\end{array}]} S(-2)\oplus S(-1)^2 \xleftarrow{\left[\begin{array}{ccc}x_2 & x_4 & 0\\ -x_1x_3 & 0 & x_4 \\ 0 & -x_1x_3 & -x_2 \end{array}\right]} &&\\ S(-3)^2 \oplus S(-2) 
 \xleftarrow{\left[\begin{array}{c}x_4\\-x_2\\x_1x_3\end{array}\right]} S(-4) \longleftarrow  0 && 
\end{eqnarray*}
The Betti number $\beta_{i,\sigma}(S/I_{\Delta^*})$ is the dimension of the module in multidegree $\sigma$ at step $i$ of the resolution, where $S/I_{\Delta^*}$ is step 0 and the steps increase as we move from left to right.  At step 0, the total degree is always 0.  For the above resolution, the multidegrees at $S(-2)\oplus S(-1)^2$ (step 1) are 1010, 0100, and 0001; at $S(-3)^2\oplus S(-2)$ (step 2), we have 1110, 1011, and 0101; and at $S(-4)$ (step 4) the multidegree is 1111.  This immediately gives us the nonzero Betti numbers:
$$\begin{array}{cccc} 
\beta_{0,0000}(S/I_{\Delta^*}) = 1, & \beta_{1,1010}(S/I_{\Delta^*}) = 1, &
\beta_{1,0100}(S/I_{\Delta^*}) = 1, & \beta_{1,0001}(S/I_{\Delta^*}) = 1, \\
\beta_{2,1110}(S/I_{\Delta^*}) = 1, &  \beta_{2,1011}(S/I_{\Delta^*}) = 1, &
 \beta_{2,0101}(S/I_{\Delta^*}) = 1, & \beta_{3,1111}(S/I_{\Delta^*}) = 1. 
 \end{array}$$
 Recalling from equation~\eqref{eq:M_H-betti} that the multidegrees correspond to complements $\bar\sigma$ of faces in $\Delta$, we can now immediately read off the elements of $\M_H(\Delta)$ from the above $\beta_{i,\bar\sigma}$ for $i>0$ as:
  $$\M_{H}(\Delta) = \{0101, 1011, 1110, 0001, 0100, 1010, 0000\} = \{24, 134, 123, 4, 2, 13, \emptyset\}.$$
\end{example}

Note that the first three elements of $\M_H(\Delta)$ above, obtained from the Betti numbers $\beta_{1,*}$ in step 1 of the resolution, are precisely the facets of $\Delta$.  The next three elements, $0001, 0100,$ and $1010,$ are mandatory codewords: they must be included for a code with simplicial complex $\Delta$ to be convex.  These all correspond to pairwise intersections of facets, and are obtained from the Betti numbers $\beta_{2,*}$ at step 2 of the resolution; this is consistent with the fact that the corresponding links are all disconnected, resulting in non-trivial $\widetilde{H}_0(\Lk_\sigma(\Delta), \kk)$.  The last element, $0000$, reflects the fact that $\Lk_{\emptyset}(\Delta) = \Delta,$ and $\dim \widetilde{H}_1(\Delta, \kk) = 1$ for $\Delta = $ L25.  By convention, however, we always include the all-zeros codeword in our codes (see Section~\ref{sec:prelim}).

Using {\tt Macaulay2} \cite{M2}, the Betti numbers for the simplicial complex $\Delta$ above can be computed through the following sequence of commands (choosing $\kk = \ZZ_2$, and suppressing outputs except for the Betti tally at the end):

\begin{verbatim}
i1 : kk = ZZ/2;
i2 : S = kk[x1,x2,x3,x4, Degrees => {{1,0,0,0},{0,1,0,0},{0,0,1,0},{0,0,0,1}}];
i3 : I = monomialIdeal(x1*x2*x4,x2*x3*x4);
i4 : Istar = dual(I);
i5 : M = S^1/Istar;
i6 : Mres = res M;      [comment: this step computes the minimal free resolution]
i7 : peek betti Mres
o7 = BettiTally{(0, {0, 0, 0, 0}, 0) => 1}
                (1, {0, 0, 0, 1}, 1) => 1
                (1, {0, 1, 0, 0}, 1) => 1
                (1, {1, 0, 1, 0}, 2) => 1
                (2, {0, 1, 0, 1}, 2) => 1
                (2, {1, 0, 1, 1}, 3) => 1
                (2, {1, 1, 1, 0}, 3) => 1
                (3, {1, 1, 1, 1}, 4) => 1
\end{verbatim}
Each line of the BettiTally displays $(i, \{ \sigma \}, |\sigma|) \Rightarrow \beta_{i,\sigma}$.  This yields (in order):
$$\beta_{0,0000} = 1,\;\; \beta_{1,0001} = 1,\;\; \beta_{1,0100} = 1,\;\; \beta_{1,1010} = 1,\;\; 
\beta_{2,0101} = 1, \;\; \beta_{2,1011} = 1, \;\; \beta_{2,1110} = 1, \;\; \beta_{3,1111} = 1,$$
which is the same set of nonzero Betti numbers we previously obtained.
Recalling again that the multidegrees correspond to complements $\bar\sigma$ in~\eqref{eq:M_H-betti}, and we care only about $i>0$, this output immediately gives us $\M_H(\Delta)$ -- exactly as before.
\bigskip

The above example illustrates how computational algebra can help us to determine whether a code has local obstructions.  However, as noted in Section~\ref{sec:examples}, even codes without local obstructions may fail to be convex.  Though we have made significant progress via Theorem~\ref{thm:Thm1}, finding a complete combinatorial characterization of convex codes is still an open problem.

\section*{Acknowledgments}
This work began at a 2014 AMS Mathematics Research Community,  ``Algebraic and Geometric Methods in Applied Discrete Mathematics,'' which was supported by NSF DMS-1321794.
CC was supported by NSF DMS-1225666/1537228, NSF DMS-1516881, and an Alfred P. Sloan Research Fellowship; EG was supported by NSF DMS-1304167 and NSF DMS-1620109; JJ was supported by NSF DMS-1606353; and AS was supported by NSF DMS-1004380 and NSF DMS-1312473/1513364.  We thank Joshua Cruz, Chad Giusti, Vladimir Itskov, Carly Klivans, William Kronholm, Keivan Monfared, and Yan X Zhang for numerous discussions, and Eva Pastalkova for providing the data used to create Figure~\ref{fig:place-fields}.


\bibliographystyle{unsrt}
\bibliography{convexity-refs}


\pagebreak
\section{Supplementary Text}

{\bf Contents}
\begin{itemize}
\item[S1] Classification of convex codes for $n \leq 4$
\item[S2] Alexander duality, the Stanley-Reisner ideal, and Hochster's formula
\item[S3] Bounds on the minimal embedding dimension of convex codes
\end{itemize}

\subsection*{S1 Classification of convex codes for $n \leq 4$}

For $n=1$ or $n=2$, all codes are convex.  The first non-convex codes appear for $n = 3$.  
Using our convention that all codes include the all-zeros codeword, there are a total of 40 permutation-inequivalent codes on $3$ neurons \cite{neural_ring}.  Of these, only 6 are non-convex (see Table~\ref{table:n3} and Lemma~\ref{lemma:n3}).
\bigskip

\begin{table}[!h]
\begin{center}
\begin{tabular}{c | l | l }
 label & code $\C$ & $\Delta(\C)$ \\
\specialrule{.1em}{0em}{.1em} 
B3 & 000, 010, 001, 110, 101 & L6  \\ 
\hline
B5 & 000, 010, 110, 101  & L6 \\
\hline
B6 & 000, 110, 101  & L6 \\
\hline
E2 & 000, 100, 010, 110, 101, 011 & L7 \\
\hline
E3 & 000, 100, 110, 101, 011 & L7 \\
\hline
E4 & 000, 110, 101, 011 & L7\\ 
\hline
\end{tabular}
\caption{All non-convex codes on $n=3$ neurons, up to permutation equivalence.  Code labels are the same as in \cite{neural_ring}, and simplicial complex labels are as in Figures~\ref{fig:all-links} and~\ref{fig:n4}.}
\label{table:n3}
\end{center}
\end{table}

\vspace{-.2in}
\begin{lemma}\label{lemma:n3}
There are 6 non-convex codes on $n \leq 3$ neurons, up to permutation equivalence.  They are the codes shown in Table~\ref{table:n3}.  
\end{lemma}

\begin{proof} 
First, we show that all six codes in Table~\ref{table:n3} are not convex.  Let $\Delta_{\text{L6}}$ and $\Delta_{\text{L7}}$ be the simplicial complexes labeled L6 and L7 in Figures~\ref{fig:all-links} and~\ref{fig:n4}.  It is easy to see that $\{1\} \in \C_{\min}(\Delta_{\text{L6}})$ and $\{1\},\{2\},\{3\} \in \C_{\min}(\Delta_{\text{L7}})$, because the corresponding links are non-contractible.  
Codes B3, B5, and B6 all have simplicial complex $\Delta_{\text{L6}}$, but are missing the codeword 100, corresponding to $\{1\}$.  Since $\{1\} \in \C_{\min}(\Delta_{\text{L6}})$, by Theorem~\ref{thm:Thm1} these codes each have a local obstruction and thus cannot be convex.  
Codes E2, E3, and E4 all have simplicial complex  $\Delta_{\text{L7}}$, but are missing the codeword 001, corresponding to $\{3\}$.  Since $\{3\} \in \C_{\min}(\Delta_{\text{L7}})$,  these codes cannot be convex.  
All remaining codes for $n = 3$ neurons were shown to be convex in \cite{neural_ring}, via explicit convex realizations in two dimensions.
\end{proof}

We now consider $n \leq 4$.  Figure~\ref{fig:n4} displays all simplicial complexes on $n \leq 4$ vertices, up to permutation equivalence, and highlights all intersections of two or more facets.  By inspection, we see that every link corresponding to a non-empty max intersection is \underline{not} contractible.   We thus have the following lemma:

\begin{lemma}\label{lemma:n4}
Let $\C$ be a neural code on $n \leq 4$ neurons.  Then $\C$ has no local obstructions if and only if $\C$ is max $\cap$-complete.
\end{lemma}

\begin{proof}
Let $\Delta = \Delta(\C)$ be the simplicial complex on $n \leq 4$ vertices corresponding to the code $\C$.  Recall from Theorem~\ref{thm:Thm1} that $\C_{\min}(\Delta)$ consists of all $\sigma \in \Delta$ such that $\Lk_\sigma(\Delta)$ is not contractible.  Since, for $n \leq 4$, all non-empty max intersections $\sigma \in \mathcal{F}_\cap(\Delta)$ have non-contractible links, it follows that $\C_{\min}(\Delta) = \mathcal{F}_\cap(\Delta)$.  Therefore, by Theorem~\ref{thm:Thm1}, $\C$ has no local obstructions if and only if $\C \supseteq \mathcal{F}_\cap(\Delta)$.
\end{proof}

Note that the proof of Lemma~\ref{lemma:n4} relied on showing that $\C_{\min}(\Delta) = \mathcal{F}_\cap(\Delta)$.  More generally, for any code $\C$ whose simplicial complex $\Delta$ satisfies $\C_{\min}(\Delta) = \mathcal{F}_\cap(\Delta)$, we know that $\C$ has no local obstructions if and only if it is max $\cap$-complete.  Combining this with Proposition~\ref{prop:max-int}, we have the following:

\begin{proposition}\label{prop:cmin-fcap}
Let $\Delta$ be a simplicial complex satisfying $\C_{\min}(\Delta) = \mathcal{F}_\cap(\Delta)$.  Then for any $\C$ with $\Delta(\C) = \Delta$, $\C$ is convex if and only if it is max $\cap$-complete.
\end{proposition}

\noindent As a corollary of Lemma~\ref{lemma:n4} and Proposition~\ref{prop:cmin-fcap}, we obtain Proposition~\ref{prop:n4}.

In fact, for $n \leq 4$ we have generated explicit convex realizations for all max $\cap$-complete codes.  Figure~\ref{fig:convex-realizations} illustrates convex realizations for the max $\cap$-complete codes corresponding to most of the simplicial complexes L1-L28.  The max $\cap$-complete codes for the omitted complexes L1-L5, L9-L10, and L12 all have obvious convex realizations, while those for L15 and L16 are obvious given the realizations for L7 and L8, respectively.  

Table~\ref{table:n4} summarizes our classification of $n \leq 4$ codes and provides minimal embedding dimensions $d(\C)$ for all the convex codes.  Note that $d(\C)$ sometimes varies for codes with the same simplicial complex, depending on which optional codewords are included.

\begin{figure}
\begin{center}
\includegraphics[width = 5.25in]{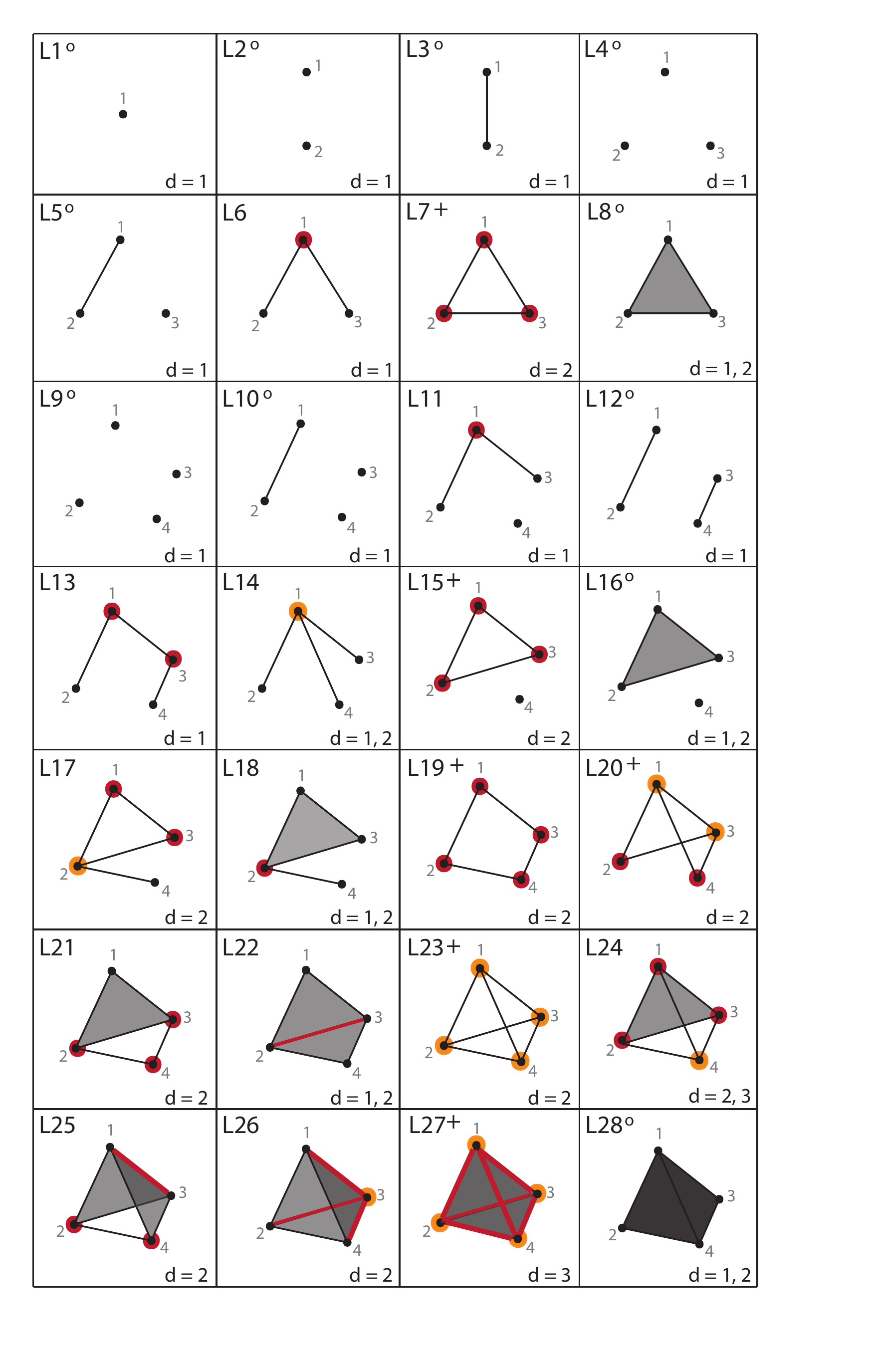}
\end{center}
\vspace{-.5in}
\caption{\small Classification of convex codes for $n \leq 4$.  L1-L28 are the 28 possible simplicial complexes $\Delta = \Delta(\C)$ that can arise, up to permutation equivalence.  For each simplicial complex, faces highlighted in color correspond to pairwise (red) and triple (orange) intersections of facets (maximal faces).  One can check that for each colored face $\sigma$, $\Lk_\sigma(\Delta)$ is not contractible.  By Proposition~\ref{prop:n4}, any code $\C$ with simplicial complex $\Delta$ is convex if and only if $\C$ includes all codewords corresponding to the colored faces.  All other non-maximal faces are optional codewords, whose inclusion or exclusion does not affect convexity.  
Simplicial complexes with no optional codewords are labeled +, while those with no (non-maximal) mandatory codewords are labeled $\circ$.  When both + and $\circ$ labels apply, as in L1, L2, and L4, we simply use $\circ$.
The possible minimal embedding dimensions $d = d(\C)$ that can arise for convex codes are shown in the lower right corners.} 
\label{fig:n4}
\end{figure}

\begin{table}
\begin{footnotesize}
\begin{center}
\begin{tabular}{c|c|l|l|l|c|l|l}
 \multirow{3}{*}{$\Delta$} & \# non-- & maximal & mandatory & optional  &  \multirow{3}{*}{$d(\C)$}  &  \multirow{3}{*}{picture} &\multirow{3}{*}{notes}\\
 & convex & codewords & codewords & codewords   &   & & \\ 
 & codes & (facets) & (non-maximal) & in $\Delta$ & & &\\
 \specialrule{.2em}{0em}{0em} 
L1$^\circ$ & 0 & 1 & none & none & 1 & &   \\
\hline
L2$^\circ$ & 0 & 10, 01 & none & none & 1 & &  \\
\hline
L3$^\circ$ & 0 & 11 & none & all (2) & 1 & &   \\
\hline
L4$^\circ$ & 0 & 100, 010, 001 & none & none & 1 & &  \\
\hline
L5$^\circ$ & 0 & 110, 001 & none & all (2) & 1 & & \\
\hline
L6 & 3 & 110, 101 & 100 & 010, 001 & 1  & L6-series & \\
\hline
L7$^+$ & 3 & 110, 101, 011 & all (3) & none & 2 & L7-series & convex $\Leftrightarrow \C = \Delta(\C)$\\
\hline
L8$^\circ$ &  0 & 111 & none  & all (6) &  1, 2 & L8-series & $d=2$ $\Leftrightarrow$ $\star$\\ 
\hline
\multirow{2}{*}{L9$^\circ$} & \multirow{2}{*}{0} & 1000, 0100 & \multirow{2}{*}{none} & \multirow{2}{*}{none}  & \multirow{2}{*}{1} & &  \\
& & 0010, 0001 & & & & &\\
\hline
L10$^\circ$ & 0 & 1100, 0010, 0001 & none & all (2) & 1 &  & \\
\hline
L11 & 3 &1100, 1010,  0001 & 1000 & 0100, 0010 & 1 & L6-series & \\
\hline
L12$^\circ$ & 0 & 1100, 0011 & none & all (4) & 1 &  &\\
\hline
L13 & 7 & 1100, 1010, 0011 & 1000, 0010 & 1000, 0001 & 1 & L6-series & \\
\hline
L14 & 4 & 1100, 1010, 1001  & 1000 & 0100, 0010, 0001 & 1, 2 & L6-series & $d=2$ $\Leftrightarrow$ $\star$  \\
\hline
\multirow{2}{*}{L15$^+$} & \multirow{2}{*}{3} & 1100, 1010 & \multirow{2}{*}{all (3)} & \multirow{2}{*}{none} & \multirow{2}{*}{2} & \multirow{2}{*}{like L7} & \multirow{2}{*}{convex $\Leftrightarrow \C = \Delta(\C)$}\\
& & 0110, 0001 & & & & &\\
\hline
L16$^\circ$ & 0 & 1110, 0001 & none & all (6) & 1, 2 & like L8 & $d=2$ $\Leftrightarrow$ $\star$  \\
\hline
\multirow{2}{*}{L17} & \multirow{2}{*}{10} & 1100, 1010 & 1000, 0100 & \multirow{2}{*}{0001} & \multirow{2}{*}{2} & \multirow{2}{*}{L7-series} &\\
& & 0110, 0101 & 0010 &  & & &\\
\hline
\multirow{3}{*}{L18} & \multirow{3}{*}{40} & \multirow{3}{*}{1110, 0101} & \multirow{3}{*}{0100} &  &  \multirow{3}{*}{1, 2} & \multirow{3}{*}{L8-series} & 
$d=2 \Leftrightarrow$ $\C$ contains \\ 
 & & & &  1000, 0010, 0001 & & & $\{1000,0010\}$, or \\
 & & & & 1100, 1010, 0110 & & & $\{1100, 0110, 1010\}$\\
 \hline
\multirow{2}{*}{L19$^+$} & \multirow{2}{*}{5} & 1100, 1010 & \multirow{2}{*}{all (4)} & \multirow{2}{*}{none} & \multirow{2}{*}{2} & \multirow{2}{*}{L6-series} & \multirow{2}{*}{convex $\Leftrightarrow \C = \Delta(\C)$} \\
&& 0101, 0011 &&&&&\\
\hline
\multirow{2}{*}{L20$^+$} & \multirow{2}{*}{8} & 1100, 1010, 1001 & \multirow{2}{*}{all  (4)} & \multirow{2}{*}{none} & \multirow{2}{*}{2} & \multirow{2}{*}{L7-series} & \multirow{2}{*}{convex $\Leftrightarrow \C = \Delta(\C)$} \\
&& 0110, 0011 &&&&&\\
\hline
\multirow{2}{*}{L21} & \multirow{2}{*}{68} & \multirow{2}{*}{1110, 0101, 0011} & 0100, 0010 & 1000, 1100 & \multirow{2}{*}{2} &\multirow{2}{*}{L8-series} &\\
& &  & 0001 & 1010, 0110  & & &\\
\hline
 \multirow{3}{*}{L22} & \multirow{3}{*}{62} & \multirow{3}{*}{1110, 0111} & \multirow{3}{*}{0110} & 1000, 0100, 0010 &\multirow{3}{*}{1, 2} & \multirow{3}{*}{L22-series} &  $d=2$ $\Leftrightarrow$ $\star$,\\
 & & & & 0001, 1100, 1010 & & & or $\{0101,0011\}\subset \C$,\\
 & & & & 0101, 0011 & & & or  $\{1100, 1010\}\subset \C$\\
 \hline
\multirow{2}{*}{L23$^+$} & \multirow{2}{*}{4} & 1100, 1010, 1001 & \multirow{2}{*}{all (4)} & \multirow{2}{*}{none} & \multirow{2}{*}{2} & \multirow{2}{*}{L7-series} & \multirow{2}{*}{convex $\Leftrightarrow \C = \Delta(\C)$} \\
& & 0110, 0101, 0011 & & & & &\\
\hline
\multirow{2}{*}{L24} & \multirow{2}{*}{36} & 1110, 1001 & 1000, 0100 & \multirow{2}{*}{1100, 1010, 0110} &   \multirow{2}{*}{2, 3} & L8-series & $d=3$ $\Leftrightarrow$ no \\
& & 0101, 0011 & 0010, 0001 & & & d=3 codes & optional codewords \\
\hline
\multirow{2}{*}{L25} & \multirow{2}{*}{168} & \multirow{2}{*}{1110, 1011, 0101} & 0100, 0001 & 1000, 0010, 1100 & \multirow{2}{*}{2} & \multirow{2}{*}{L22-series} & \\
& &  & 1010 & 1001, 0110, 0011& & & \\
\hline
\multirow{2}{*}{L26} & \multirow{2}{*}{407} & \multirow{2}{*}{1110, 1011,  0111} & 0010, 1010 & 1000, 0100, 0001 & \multirow{2}{*}{2} & \multirow{2}{*}{L22-series} &\\
& & & 0110, 0011 & 1100, 1001, 0101 & & &\\
\hline
\multirow{2}{*}{L27$^+$} & \multirow{2}{*}{85} & 1110, 1101 & \multirow{2}{*}{all (10)} & \multirow{2}{*}{none} & \multirow{2}{*}{3} & \multirow{2}{*}{d=3 codes} & \multirow{2}{*}{convex $\Leftrightarrow \C = \Delta(\C)$}\\
&& 1011, 0111 &&&&&\\
\hline
L28$^\circ$ & 0 & 1111 & none & all (14) & 1, 2 & L22-series & $d=2$ $\Leftrightarrow$ $\star$\\
\specialrule{.1em}{0em}{0em} 
\end{tabular}
\end{center}
\end{footnotesize}
\caption{\small Convexity and dimension for codes on $n \leq 4$ neurons.  For each simplicial complex $\Delta$, labeled as in Figure~\ref{fig:n4}, the second column is the number of non-convex codes $\C$ such that $\Delta(\C) = \Delta$, up to permutation equivalence and including the all-zeros codeword, while the sixth column $d(\C)$ displays the possible minimal embedding dimensions for convex codes only.  The third column lists the codewords corresponding to facets of $\Delta$; these are automatically included in any code with simplicial complex $\Delta$.
The fourth columns gives all other non-empty mandatory codewords -- that is, elements of $\C_{\min}(\Delta)$ that are not facets of $\Delta$.
Optional codewords are elements of $\Delta$ whose presence or absence does not affect whether or not a code is convex, though they may alter the minimal embedding dimension $d(\C)$.  When all non-maximal codewords are mandatory or all are optional, their total number is given in parentheses.   The picture column indicates the groupings used for the convex realizations in Figure~\ref{fig:convex-realizations}.
In the notes column, $\star$ indicates that the set of optional codewords in $\C$ can {\it not} form a 2-chain.
A collection of codewords forms a {\it chain}  if we can completely order the respective sets by containment -- so  $\{1111, 1100, 1000\}$ is a chain, but $\{1110, 1000, 1101\}$ is not.  A collection of codewords can form a {\it 2-chain} if it can be partitioned into two sets (possibly empty) which are both chains.  + and $\circ$ are the same as in Figure~\ref{fig:n4}.}
\label{table:n4}
\end{table}

\begin{figure}
\begin{center}
\includegraphics[width=6in]{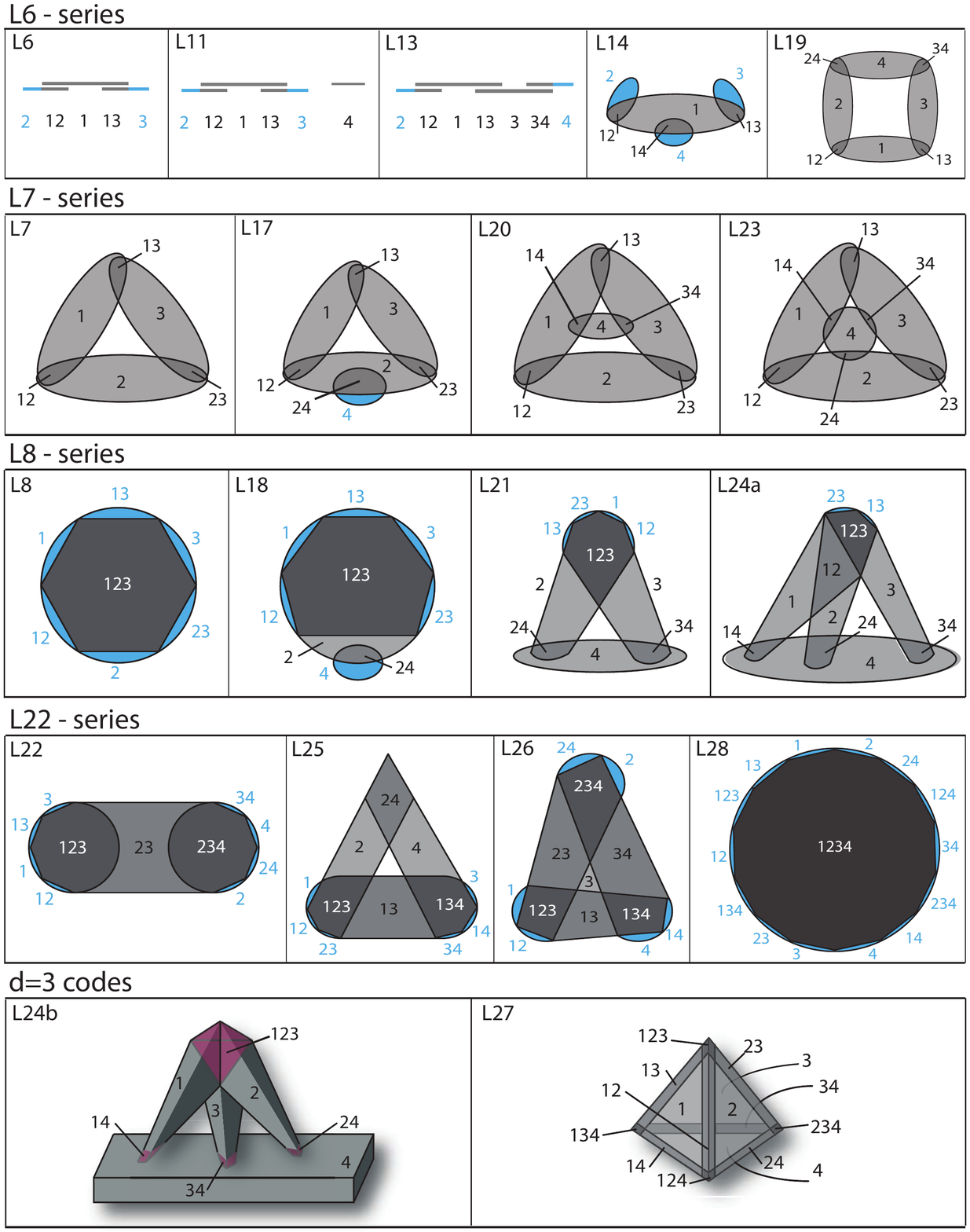}
\end{center}
\vspace{-.25in}
\caption{Convex realizations for codes on $n \leq 4$ neurons.  Each convex set $U_i$, for neuron $i$, is the union of all regions corresponding to codewords containing $i$ (see also Figure~\ref{fig:ngon}).  Note that each picture displays regions corresponding to mandatory codewords in various shades of gray, while optional codewords are in blue.   A single picture thus shows convex realizations for all convex codes corresponding to the same simplicial complex, as blue regions can be included or excluded without affecting convexity.
}
\label{fig:convex-realizations} 
\end{figure}

\FloatBarrier

\subsection*{S2 Alexander duality, the Stanley-Reisner ideal, and Hochster's formula}\label{sec:Hochster}

For any simplicial complex $\Delta$ on vertex set $[n]$, the {\em Alexander dual} is the related simplicial complex:
$$\Delta^* \od \{\bar\tau \mid \tau \notin \Delta\},$$
where $\bar\tau = [n] \setminus \tau$ denotes the complement of $\tau$ in $[n]$.
Note that $(\Delta^*)^* = \Delta$. 
Any $\Delta$ also has an associated ideal known as the \emph{Stanley-Reisner ideal}:
$$I_\Delta \od \left\{\prod_{i \in \sigma} x_i~|~ \sigma \notin \Delta\right\}.$$

Alexander duality relates the reduced homology of a simplicial complex to the cohomology of its Alexander dual,
$$\widetilde{H}_{i}(\Delta^*, \kk) \cong \widetilde{H}^{n-i-3}(\Delta, \kk),$$
where $\kk$ is a field.  Meanwhile, Hochster's formula relates the nonzero Betti numbers from a minimal free resolution of the Stanley-Reisner ideal to the reduced cohomology of restricted simplicial complexes \cite{miller-sturmfels}.  Specifically,
\begin{equation}\label{Hochster}
\beta_{i-1,\sigma}(I_\Delta) = \beta_{i,\sigma}(S/I_\Delta) = \dim \widetilde{H}^{|\sigma| - i - 1}(\Delta|_{\sigma}, \kk),
\end{equation}
where $S = \kk[x_1,\ldots,x_n]$ and the $\beta_{i,\sigma}$s refer always to Betti numbers for a {\em minimal} free resolution.\footnote{See \cite[Chapter 1]{miller-sturmfels} for more details about free resolutions of monomial ideals.}  

The following link lemma connects Hochster's formula~\eqref{Hochster} to its dual version, below.  The dual formulation is more useful to us, as it allows us to compute the dimensions of all non-trivial homology groups for all links, $\Lk_\sigma(\Delta)$, from a single free resolution.

\begin{lemma}\label{lemma:pre-Hoch}
$\Lk_\sigma(\Delta) = (\Delta^*|_{\bar\sigma})^*.$
\end{lemma}

\begin{proof}
First, observe that $\Delta^*|_{\bar \sigma} = \{\tau \mid \tau \subset \bar\sigma \text{ and } \tau \in \Delta^*\}.$  The dual is thus
$(\Delta^*|_{\bar\sigma})^* = \{\bar\sigma \setminus \tau \mid \tau \subset \bar\sigma \text{ and } \tau \not\in \Delta^*\}
= \{ \omega \mid \omega \subset \bar\sigma \text{ and } \bar\sigma \setminus \omega \notin \Delta^*\}
= \{ \omega \mid \omega \cap \sigma = \emptyset \text{ and } \sigma \cup \omega \in \Delta\} = \Lk_\sigma(\Delta).$
\end{proof}

\begin{lemma}[{Hochster's formula, dual version, cf. \cite[Corollary 1.40]{miller-sturmfels}}]\label{lemma:dual-Hoch}
$$\dim \widetilde{H}_i(\Lk_\sigma(\Delta), \kk) = \beta_{i+2,\bar\sigma}(S/I_{\Delta^*}).$$
\end{lemma}

\begin{proof}
Using (in order) Lemma~\ref{lemma:pre-Hoch}, Alexander duality, and the original version of Hochster's formula~\eqref{Hochster}, we obtain:
$\dim \widetilde{H}_i(\Lk_\sigma(\Delta), \kk) = \dim \widetilde{H}_i((\Delta^*|_{\bar\sigma})^*, \kk) = \dim \widetilde{H}^{|\bar\sigma|-i-3}(\Delta^*|_{\bar\sigma}, \kk)
= \beta_{i+2,\bar\sigma}(S/I_{\Delta^*}).$
\end{proof}

It is important to note that if $\sigma$ is a facet of $\Delta$, then $\Lk_\sigma(\Delta) = \emptyset,$ which is non-contractible due to nontrivial homology in degree $-1$.  Hochster's formula thus detects facets of $\Delta$ via the nonzero Betti numbers $\beta_{1,\bar\sigma}$, as these correspond to $\sigma$ such that $\dim \widetilde{H}_{-1}(\Lk_\sigma(\Delta), \kk) > 0$.  Note also that $\Lk_{\emptyset}(\Delta) = \Delta$, so  
if $\Delta$ itself has nontrivial reduced homology in degree $i$, this will be detected as a nonzero Betti number $\beta_{i+2,[n]}$, where $\bar\sigma = [n]$ is the complement of $\sigma = \emptyset$.

\subsection*{S3 Bounds on the minimal embedding dimension of convex codes}
 
We now turn to the problem of determining the {\it minimal embedding dimension} $d(\C)$ of a convex code $\C$.  
There is no general method for computing $d(\C)$, though bounds can be obtained from the information present in the simplicial complex $\Delta(\C)$.  In this section, we review known results on $d$-representability and Helly's theorem, and apply them to obtain lower bounds on $d(\C)$.  We also obtain an additional bound on $d(\C)$ from the Fractional Helly theorem, and examine how it compares to the Helly's theorem bound.

Our dimension bounds all rely solely on features of the code captured by $\Delta(\C),$ and do not take into account the finer structure of the code.  Nevertheless, the presence or absence of a single codeword can have a significant effect on $d(\C),$ even if the simplicial complex $\Delta = \Delta(\C)$ is fixed (e.g., see Table~\ref{table:n4} for L8, L14, L16, etc.).  It remains an open question how to use this additional information in order to improve the bounds on $d(\C)$.

\subsection*{Embedding dimension and $d$-representability}\label{sec:d-representability}
The problem of determining $d(\C)$ for a convex code $\C$ has not been directly addressed in the literature.  However, the related problem of determining when a simplicial complex $\Delta$ can be realized as the nerve $\N(\U)$ of a cover $\U$ has received considerable attention (see \cite{tancer-survey, Tancer-d-rep} and references therein).  A simplicial complex $\Delta$ is said to be \emph{$d$-representable} if there exists a collection of convex (not necessarily open) sets $\U=\{U_1, \ldots, U_n\}$, with $U_i \subset \R^d,$ such that $\Delta=\N(\U)$.   Note that for such a $\U$, the corresponding code $\C(\U)$ need not be equal to $\Delta$, though it is always true that $\Delta(\C(\U)) = \Delta$.

While $d$-representability of $\Delta(\C)$ does not tell us the value of $d(\C)$, it does provide a lower bound.
This motivates us to define the \emph{nerve dimension} $d_\N(\C)$ of a code $\C$ to be the minimal $d$ such that $\Delta(\C)$ is $d$-representable.  It immediately follows that
\begin{equation*}\label{ineq:nervedim}
d(\C) \geq d_\N(\C), 
 \end{equation*}
 because any embedding of $\C$ in $\R^{d(\C)}$ via a collection of convex open sets $\U$ is also a realization of $\Delta(\C)$ as the nerve $\N(\U)$.  Unfortunately, $d_\N(\C)$ may be difficult to compute in general.  In contrast, we can obtain a lower bound from Helly's theorem that is simple to read off from the structure of $\Delta(\C)$.

\subsection*{Bounds from Helly's theorem}
One common tool used to address $d$-representability of simplicial complexes is Helly's theorem.
\begin{lemma}[Helly's theorem \cite{helly}]
Let $\U=\{U_1, \ldots, U_n\}$ be a collection of convex open sets in $\R^d$.  If for every $d+1$ sets in $\U$, the intersection is non-empty, then the full intersection $\bigcap_{i=1}^n U_i$ is also non-empty.
\end{lemma}

Helly's theorem implies that if $\Delta$ is $d$-representable and $\Delta$ contains all possible $d$-dimensional faces, then $\Delta$ must be the full simplex.  On the other hand, if $\Delta$ contains all possible $d$-dimensional faces but is not the full simplex, then it is {\it not} $d$-representable.  This immediately yields examples where
the presence or absence of a single codeword can have a large effect on $d(\C)$.    

\begin{proposition}\label{prop:gap}
Let $\C$ be a code on $n$ neurons, and suppose that for some $k$ with $1 \leq k <n$, $\Delta(\C)$ contains all $k$-dimensional faces.  
If $11\cdots1 \in\C$, then $d(\C)\leq 2$; otherwise, $d(\C) > k$.
\end{proposition}

\begin{proof}
In the first case, where $11\cdots1 \in \C$, the fact that $\C$ is convex and $d(\C) \leq 2$ follows from Lemma~\ref{lemma:D2}.
For the second case, where $11\cdots1 \notin \C$, suppose $\C$ is realizable as a convex code in $\mathbb{R}^d$ for some $d\leq k$, so that $\C = \C(\mathcal{U})$ for some collection of convex open sets $\mathcal{U}=\{U_1,U_2,\ldots,U_n\},$ with each $U_i \subset \RR^d$.  Since, by hypothesis, $\Delta(\C)$ contains all $k$-dimensional faces, it also contains all $d$-dimensional faces, and so the intersection of every collection of $d+1$ subsets in $\mathcal{U}$ is non-empty.  Thus, by Helly's Theorem, the full intersection of all sets in $\mathcal{U}$ is non-empty, and so $11\cdots1 \in \C$.  This contradicts the fact that $11\cdots1 \notin \C$; hence, we must have $d(\C) > k$.  
\end{proof}

We can also apply Helly's theorem to every subcollection $\{ U_{i_1}, \ldots, U_{i_m}\} \subset \U$, or equivalently to the induced subcomplex on elements $i_1, \ldots, i_m$, to see that if all the $d$-dimensional faces of this subcomplex are present, then the top-dimensional face must also be present in order for $\Delta$ to be $d$-representable.  This leads us to the following definitions.
A simplicial complex is said to contain an \emph{induced $k$-dimensional simplicial hole} if it contains $k+1$ vertices such that the induced subcomplex on those vertices is isomorphic to a hollow simplex (see Section~\ref{sec:prelim}, and \cite{tancer-survey}).   We define the \emph{Helly dimension}\footnote{The closely-related notion of \emph{Helly number} for a simplicial complex was previously introduced in the literature.  Specifically, the Helly number of $\Delta(\C)$ is $d_H(\C) + 1$.} of $\C$, denoted $d_H(\C)$, to be the dimension of the largest induced simplicial hole of $\Delta(\C)$:
\begin{equation}\label{eq:d_H}
d_H(\C) \od  \max \{k \mid \Delta(\C) \text{ has a $k$-dimensional induced simplicial hole} \}.
\end{equation}
Clearly, $d_H(\C) \leq d(\C)$.

\subsection*{Bounds from the Fractional Helly theorem}
The Fractional Helly theorem is a well-known extension of Helly's theorem that provides new bounds on $d(\C)$, though they are not always better.

\begin{lemma}[Fractional Helly theorem\footnote{In \cite{kalai}, Kalai gave a sharp version of Theorem \ref{thm:FractionalHelly}: If $\alpha \binom{n}{d+1}$ of the $(d+1)$-tuples have non-empty intersections, then there exists $\sigma$ with $|\sigma | > (1-(1-\alpha)^{1/(d+1)})n$ such that the corresponding intersection is non-empty.  The precise value of $\beta$ for which $|\sigma| > \beta n$ is not essential to the results in the remainder of our paper. Therefore, for ease of notation and computation, we prefer $\beta  = \alpha/(d+1)$ as in Theorem \ref{thm:FractionalHelly}, while keeping in mind that a tighter bound exists.}, Theorem 6.7 of~\cite{tancer-survey}]\label{thm:FractionalHelly}
Let $\alpha>0$, and $\U=\{U_1,U_2,\ldots,U_n\}$ be a collection of convex open sets in $\RR^d$ such that at least $\alpha \binom{n}{d+1}$ of the $(d+1)$-tuples of sets in $\U$ have non-empty intersections.  Then there exists $\sigma \subseteq [n]$ such that $|\sigma | > \frac{\alpha}{d+1}n$, and $\bigcap_{i \in \sigma} \, U_i$ is non-empty.
\end{lemma}

The Fractional Helly theorem indicates that if a code $\C$ can be embedded in $\mathbb{R}^d$, and the simplicial complex $\Delta(\C)$ has many $d$-dimensional faces, then $\Delta(\C)$ must have some sufficiently high-dimensional face. The following lemma quantifies these observations in our context.

\begin{lemma}\label{lem:FracHelly}
Let $\Delta$ be a $k$-dimensional simplicial complex on $n$ elements, and let $f_d(\Delta)$ be the number of $d$-dimensional faces in $\Delta$ for $1 \leq d <n$.  If $\Delta$ is $d$-representable, then $k+1 > f_d(\Delta)/\binom{n-1}{d}$.  
\end{lemma}
\begin{proof}
By definition of $d$-representable, we have $\Delta=\N(\U)$ for some $\mathcal{U} = \{U_1,U_2,\ldots,U_n\}$, where each $U_i \subset \RR^d$.  Since each $d$-dimensional face of $\Delta$ corresponds to an intersection of $(d+1)$ of the $U_i$s, we have that $f_d(\Delta(\C))$ of the $(d+1)$-tuples have non-empty intersections.  By the Fractional Helly theorem, there is some $\sigma \subseteq [n]$ with $|\sigma| > \frac{\alpha}{d+1}n$ such that $\bigcap_{i \in \sigma} U_i \neq \emptyset$, where $\alpha = \frac{f_d(\Delta(\C))}{\binom{n}{d+1}}$.  Since $\Delta(\C)$ is $k$-dimensional, it follows that $|\sigma| \leq k+1$, and so $k+1>\frac{\alpha}{d+1}n = f_d(\Delta(\C))/\binom{n-1}{d}$. 
\end{proof}

This leads us to the following definition.
Let $\C$ be a code on $n$ neurons with a $k$-dimensional simplicial complex $\Delta(\C)$, and let $f_d(\Delta(\C))$ be the number of $d$-dimensional faces in $\Delta(\C)$ for $1 \leq d <n$. The \emph{Fractional Helly dimension} $d_{FH}(\C)$ of $\C$ is given by:
\begin{equation}\label{d_FH}
d_{FH}(\C)\od 1+ \max \left \{ d \ \Big| \ f_d(\Delta(\C))\geq (k+1) \cdot \binom{n-1}{d}, \ 1 \leq d < n \right \}\,.
\end{equation}

Further dimension bounds based on the $f$-vector $\{f_i(\Delta)\}$ can be obtained from the results in \cite{kalai-fvector-1984, kalai-fvector-1986}.

\subsection*{Comparison of dimension bounds}

How do the Helly and Fractional Helly dimensions $d_H(\C)$ and $d_{FH}(\C)$ compare to each other and to the nerve dimension $d_\N(\C)$?  First, note that although a simplicial complex $\Delta$ cannot be represented in any dimension less than its Helly dimension $d_H(\Delta)$, Helly's theorem does not guarantee that $\Delta$ is $d_H$-representable.  Thus we have 
$d_H(\C) \leq d_\N(\C) \leq d(\C)\,.$
By the same reasoning, $d_{FH}(\C) \leq d_\N(\C) \leq d(\C)\,$.
The next examples show that we can have $d_H(\C) < d_\N(\C)$ and $d_{FH}(\C) < d_\N(\C)$;  i.e., the nerve dimension can provide a strictly stronger lower bound. 

\begin{figure}[h]
\begin{center}
\includegraphics[width=3.5in]{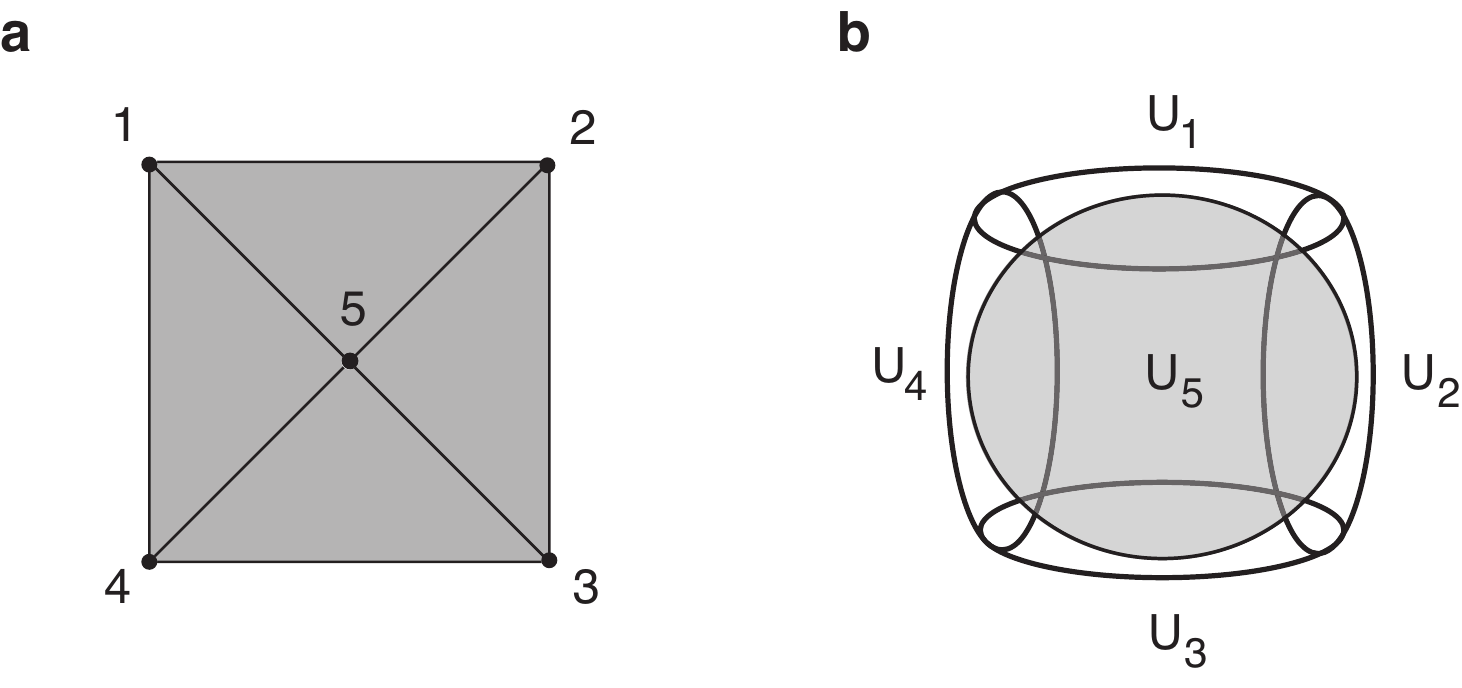}
\end{center}
\vspace{-.2in}
\caption{ {\bf(a)} A simplicial complex $\Delta$, with facets 125, 235, 345, 145.
{\bf (b)} A convex realization of $\Delta$ in $\R^2$.}
\label{fig:envelope}
\end{figure}

\begin{example}\label{ex:helly1}
Consider any code $\C$ such that $\Delta(\C)$ is the simplicial complex in Figure~\ref{fig:envelope}a.  We obtain $d_H(\C)=1$, because the two maximal induced simplicial holes of $\Delta(\C)$ arise from the subsets $\{1,3\}$ and $\{2,4\}$, which both have dimension 1.  Although $\Delta(\C)$ is contractible, the induced subcomplex on $\{1,2,3,4\}$ is a 1-cycle, so $\Delta(\C)$ is at best 2-representable and $d_\N(\C)\geq 2$.     Thus, $d_\N(\C) > d_H(\C)$.  Figure~\ref{fig:envelope}b shows that the minimal embedding dimensions is, in fact, $d(\C) = 2$.
\end{example}

Our final example shows not only that we can have $d_{FH}(\C)<d_{\N}(\C)$, but also that it is possible for $d_{FH}(\C)<d_{H}(\C)$ or for $d_H(\C)<d_{FH}(\C)$, depending on the code.  So although $d_{\N}(\C)$ is always the strongest of the three bounds, neither of the easier-to-compute $d_H(\C)$ and $d_{FH}(\C)$ bounds is universally stronger than the other.  In other words, all we can say in general is that the minimal embedding dimension satisfies:
\begin{equation}
d(\C) \geq d_{\N}(\C) \geq \max\{d_H(\C), d_{FH}(\C)\}.
\end{equation}

\begin{figure}[!h]
\begin{center}
\bigskip
\includegraphics[width=3.5in]{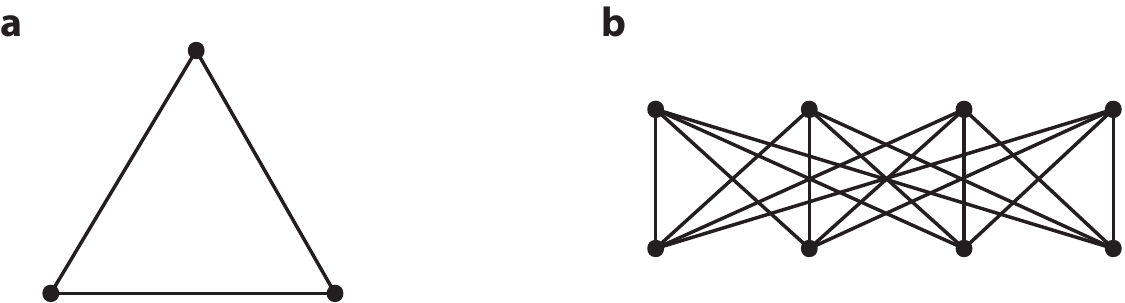}
\end{center}
\caption{ \textbf{(a)} The simplicial complex of code $\C_1$ from Example~\ref{ex:d_Hd_FH}.  \textbf{(b)} The simplicial complex of code $\C_2$ from Example~\ref{ex:d_Hd_FH}, for $r=4$.}
\label{fig:d_Hd_FH}
\end{figure}

\begin{example}\label{ex:d_Hd_FH} 
(a) Let $n=3$, and consider the code $\C_1 = \{000,110,101,011\}$, whose simplicial complex $\Delta(\C_1)$ is the empty triangle, as shown in Figure~\ref{fig:d_Hd_FH}a.  This is a hollow 2-simplex, so we have $d_H(\C_1) = 2$ and $d_\N(\C_1) = 2$.  In contrast, a quick computation yields $d_{FH}(\C_1) = 1$, and so $d_H(\C_1) = d_\N(\C_1) >d_{FH}(\C_1)$.

(b)  Let $n = 2r$, where $r \geq 4$, and suppose $\Delta(\C_2) = K_{r,r}$, the complete bipartite graph on $2r$ vertices, as shown in Figure~\ref{fig:d_Hd_FH}b. This graph contains no triangles, so the largest induced simplicial holes result from missing edges in $\Delta(\C_2)$, which have dimension 1.  Thus, $d_H(\C_2)=1$.  To compute $d_{FH}(\C_2)$, we first find the $f$-vector of $\Delta(\C_2)$.  Observe that $f_0(\Delta(\C_2))=n$, $f_1(\Delta(\C_2))=r^2$, and $f_i(\Delta(C_2))=0$ for $2 \leq i <n$.  Note also that $k=1$, since $\Delta(\C_2)$ is 1-dimensional.  Since $f_d(\Delta(\C_2)) = 0$ for $d \geq 2$, the inequality $f_d(\Delta(\C_2))\geq (k+1) \cdot \binom{n-1}{d}$ is {\em not} satisfied for these values of $d$.  However, we have
\[
f_1(\Delta(\C_2)) = r^2 \geq 2(2r-1) = 2(n-1) = (k+1) \cdot \binom{n-1}{1}\,,
\]
with the inequality being valid for all $r \geq 4$.  Thus, directly from the definition, we find $d_{FH}(\C_2)=2.$  Hence, $d_{FH}(\C_2)>d_{H}(\C_2)$.
\end{example}

\end{document}